\documentclass[12pt]{article}

\usepackage{graphicx}
\usepackage{amsmath}
\usepackage{amsfonts}
\usepackage{amssymb}
\usepackage{amsthm}
\usepackage[all,cmtip]{xy}

\usepackage{amsmath, amsthm, amssymb}

\theoremstyle{plain}	

\newtheorem{thm}{Theorem}[section]
\newtheorem{lemma}[thm]{Lemma}
\newtheorem{prop}[thm]{Proposition}
\newtheorem{cor}[thm]{Corollary}

\theoremstyle{definition}

\newtheorem{defn}[thm]{Definition}

\title{The algebraic geometry of Harper operators}

\author{
        \textsc{Dan Li} \\
        \\
	    Department of Mathematics, Florida State University \\
		Tallahassee, FL 32306 \\
		Email: dli@math.fsu.edu
	    }

\date{}

\begin{document}

\maketitle


\begin{abstract}

Following an approach developed by Gieseker, Kn\"orrer and Trubowitz for 
discretized Schr\"odinger operators, 
we study the spectral theory of Harper operators in dimension two and one, as a
discretized model of magnetic Laplacians, from the point of view of algebraic
geometry. We describe the geometry of an associated family of Bloch varieties
and compute their density of states. Finally, we also compute some spectral functions based on the density of states.

We discuss the difference between the cases with rational or irrational parameters: 
for the two dimensional Harper operator, the compactification of the Bloch variety is 
an ordinary variety in the rational case and an ind-pro-variety in the irrational case.
This gives rise, at the algebro-geometric level of Bloch varieties, to a phenomenon
similar to the Hofstadter butterfly in the spectral theory. 
In dimension two, the density of states can be expressed in terms of period integrals over Fermi curves, where the
resulting elliptic integrals are independent of the parameters. 

In dimension one, for the almost Mathieu operator, with a similar argument  
we find the usual dependence of the spectral density on the parameter,
which gives rise to the well known Hofstadter butterfly picture.

\end{abstract}

\section{Introduction}

In mathematical physics, the Hamiltonian of a lattice electron in a uniform magnetic field, called the discrete magnetic Laplacian,  has been studied for years in the tight-binding model \cite{Sh}. As a special discrete magnetic Laplacian, the Harper operator corresponds to a square lattice when the coupling constant is fixed (i.e. $\lambda = 1 $) \cite{H}. For instance, the Harper operator arises in the study of the integer quantum Hall effect \cite{BES}. One important spectral property of the Harper operator is that its spectrum is a Cantor set of zero Lebesgue measure for  every irrational frequency \cite{AK}, \cite{L}.  
In this paper we show that the different structure of the spectrum (band or Cantor-like) in
the rational and irrational case has an analog in terms of the algebro-geometric
properties of the Bloch variety, which is an ordinary variety in the rational case and
an ind-pro-variety (with properties analogous to a totally disconnected space) in the 
irrational case.  

More precisely, we consider the algebro-geometric approach to the spectral theory of
electrons in solids developed for discrete periodic Schr{\"o}dinger operators by
Gieseker, Kn\"orrer and Trubowitz in \cite{GKT} and we investigate how to extend
this approach to the case of Harper operators and almost Mathieu operators.
As we will see in detail in \S \ref{HarperAGsec} below, the main obstacle, which
so far prevented people from extending this method to the case with magnetic
field, lies in the fact that, in the irrational case, one leaves the world of algebraic
varieties, due to the simultaneous presence of infinitely many components. 
We will resolve this problem here by describing the relevant geometric space
via a suitable ``limit procedure", by which one can make sense of spaces that
are not themselves algebraic varieties, but that are built up from a sequence
of algebraic varieties. Making sense of this limit procedure requires working in
a setting where operations such as taking limits of geometric objects are possible,
and this is the reason why we have to resort to a moderate use of categorical
tools and methods, which are precisely built in order to allow for notions of limits
of the type we need. 

Our goal in this paper is to obtain explicit expressions for the spectral densities of the
two dimensional Harper operator and of its ``degeneration", the one-dimensional
almost Mathieu operator, in terms of periods on a suitable ``limit of algebraic varieties".
In order to obtain this result, we first need to overcome two main problems of a 
geometric nature:
\begin{enumerate}
\item Take a limit of algebraic varieties;
\item Compactify and remove singularities in a way that is compatible with this limit procedure.
\end{enumerate}

The first problem is the one we already mentioned:  we will see that, in the case of irrational 
parameters, the geometric space that describes the spectral problem for the 
Harper operator is the locus of zeros of a countable family of polynomial equations. 
As such it is not an algebraic variety, but we will show that it can be constructed by taking 
a {\em limit} (in a suitable sense, technically called an ind-variety) 
of a family of algebraic varieties. Roughly speaking, this means that the relevant space
can be built by progressively adding more and more components, so that one can
carry out operations such as period computations essentially ``component-wise", with
certain compatibilities taken into account.  

The second problem is more subtle. In the original formulation of \cite{GKT}, for
the case without magnetic field, one needs to compactify the Bloch variety (which
describes the spectral problem for the discretized periodic Schr{\"o}dinger operator)
by adding a Fermi curve ``at infinity" and to resolve singularities by blowups, in order
to be able to work in a smooth setting, where the period integral giving the spectral
density function is computed.  This problem occurs in exactly the same way when
we extend this method to Harper and almost Mathieu operators, but it is now further
complicated by the fact that we want to be able to perform these operations of
compactification and removal of singularities in a way that is compatible with the
``limit procedure" already mentioned, so that we can ``pass to the limit" and still 
achieve the desired result. 

We will see that this problem forces us to replace a limit with a double limit.
More precisely, on the one hand there is what one refers to as an inductive
limit, namely the operation of assembling together the countably many
components of the Bloch ind-variety. On the other hand there is the procedure
that removes singularities in each component (again, with compatibilities) and
that is achieved by {\em blowup} operations. These create varieties that are
progressively less singular and which project down to the original singular
variety, in such a way as to not change anything where the latter was already
smooth. This other operation defines what is called a projective limit and the
desired compatibility of these two limiting operations that allow us to solve
both problems (1) and (2) at the same time is encoded in what is called 
an ind-pro-variety in the current algebro-geometric terminology. Although
we work in a very special case of this very general kind of construction, it
will be convenient to adopt this general framework, as that assures us that
we can ``pass to the limit" by consistently carrying out calculations in the
intermediate steps, at the level of the successive approximations.

While this geometric construction is appropriately formulated in algebro-geometric
and categorical language where we can appropriately define the necessary double
limit operation, the resulting space we obtain has a simpler heuristic description in 
more directly physical terms. In the original case of \cite{GKT} of
the periodic Schr{\"o}dinger operator without magnetic field, the  
Bloch variety physically describes the complex energy-crystal
momentum dispersion relation, that is, the set of complex points that can be reached
via analytic continuation from the band functions.  In the case of Harper and
almost Mathieu operators with irrational parameters, one observes the phenomenon
that, instead of intervals (band structure), the spectrum takes the form of a Cantor set,
which is geometrically a suitable {\em limit} of an approximating family intervals. 
Correspondingly, the geometric space describing the complex energy-crystal
momentum dispersion relation, which replaces the Bloch variety, is no longer
an algebraic variety, but it is obtained in the same way by ``analytic continuation"
from a Cantor set rather than from bands in the spectrum. It is then natural to 
expect that such a geometric space will be the ``algebro-geometric analog of
a Cantor set", also obtained as a limit, in a way that mirrors the construction of
the Cantor set as a limit of intervals. 
It is then amusing to see that this concrete finding is in complete 
agreement with a general categorical result of Kapranov \cite{KV}, \cite{P}, which shows that
all ind-pro-varieties (which simply means all spaces that are obtained by a
``double limit procedure" of the kind described above) in fact really do look like Cantor sets.

The paper is structured as follows: the geometric problems (1) and (2) above
are dealt with in \S \ref{HarperAGsec1}--\S \ref{HarperAGsec5}. A more physics oriented reader 
who is willing to believe that a suitable double
limit of varieties needed to solve both the first and the second problem can be
constructed, can skip directly to \S \ref{DenSec1} and the following sections,
where the space obtained in
the previous part of the paper is used to perform the needed period computations that
describe the density of states and the spectral functions for the Harper and
almost Mathieu operators.  In particular, \S  \ref{DenSec1} and \S \ref{SpFunc1sec}
deal with the case of the two-dimensional Harper operator, while \S \ref{DenSec2}
deals with the one-dimensional case of the almost Mathieu operator.

\subsection{The algebro-geometric setup}

In \cite{GKT}, Gieseker, Kn\"orrer and Trubowitz modeled the theory of electron propagation in solids by considering a discrete periodic Schr{\"o}dinger operator $-\Delta + V$ acting on the Hilbert space $\ell^2(\mathbb{Z}^2)$, with 
\begin{equation}\label{Delta0}
\Delta \psi(m,n) = \psi(m+1,n) + \psi(m-1,n)  + \psi(m,n+1)  + \psi(m,n-1)
\end{equation}
the random walk operator that discretizes the Laplacian, and 
with an effective potential $V$ given by a real function, periodic with respect to a sublattice $a\mathbb{Z} + b\mathbb{Z}$, where $a$ and $b$ are distinct odd primes. They studied the geometry of the associated Bloch variety

\begin{equation}\label{BVvar}
\begin{array}{rl}
B(V) = & \{ (\xi_1, \xi_2, \lambda) \in   
\mathbb{C}^* \times \mathbb{C}^* \times \mathbb{C}  \,|\,  \\
& \exists \text{ nontrivial }   \psi  \text{ with }  (-\Delta + V) \psi =\lambda \psi,  \\
& \text{such that } \psi(m +a,n) = \xi_1 \psi(m,n) \\
& \text{and } \, \psi(m ,n +b) = \xi_2 \psi(m,n), \,  \forall ~(m,n) \in \mathbb{Z}^2 \}
\end{array}
\end{equation}

One then considers the projection
$\pi : B(V) \rightarrow \mathbb{C}$ on the third coordinate and one 
defines the Fermi curves as $F_\lambda(\mathbb{C}) := \pi^{-1}(\lambda)$.

In more physical terms, the Bloch variety represents the complex energy-crystal
momentum dispersion relation, that is, the locus of all complex points that can
be reached by analytic continuation from any band function.

Defining the integrated density of states $\rho(\lambda)$ as the averaged counting function of the eigenvalues, they observed that the density of states $d\rho/d\lambda$ can be expressed as a period integral over the homology class $ F_\lambda$, namely
$$
\frac{d\rho}{d\lambda} = \int_{F_\lambda} \omega_\lambda
$$
where $\omega_\lambda$ is a holomorphic form on the complex curve $F_\lambda(\mathbb{C})$.

The interested reader is referred to 3D VRML Fermi Surface Database \cite{CNCHS}
for concrete physical examples of Fermi curves.

\subsection{The case with magnetic field: Harper operator}

In this paper we consider a different, but closely related, spectral problem, where instead of the
usual Laplacian and its discretization given by the random walk operator, one considers a 
magnetic Laplacian, whose discretization is a Harper operator. We also restrict our attention to the case with trivial potential $V\equiv 0$.

The (two dimensional) Harper operator $H$ acting on $\ell^2(\mathbb{Z}^2)$ is defined as
\begin{equation}\label{Harper}
\begin{array}{rl}
H \psi (m,n) : = &
                              e^{-2\pi i \alpha n} \psi(m+1,n) +
                              e^{2\pi i \alpha n} \psi(m-1,n)  + \\
                           &    e^{-2\pi i \beta m} \psi(m,n+1)  +
                              e^{2\pi i \beta m} \psi(m,n-1)
\end{array}                            
\end{equation}
where the two unitaries
$$
U\psi(m,n) : =e^{-2\pi i\alpha n} \psi(m+1,n) \ \  \text{ and }  \ \  V\psi(m,n) : =e^{-2\pi
i\beta m} \psi(m,n+1) 
$$ 
are the so-called magnetic translation operators with phases $\alpha$ and $\beta$ respectively, and the group of magnetic translations $T_{\alpha, \beta}$ is generated by $U, V$. 
One can then write the Harper operator as $H= U + U^* + V + V^*$. Note that our form of the Harper operator is slightly different from that in the literature, but they are all unitary equivalent by a gauge transformation 
$T_\gamma \psi(m,n) = e^{2\pi i \gamma mn} \psi(m,n)$ on $\ell^2(\mathbb{Z}^2)$.

Let ${\mathcal T}_{\alpha,\beta}=C^*(T_{\alpha,\beta})$ be the group 
$C^*$-algebra of the group of magnetic translations $T_{\alpha,\beta}$.
Recall that the noncommutative torus $A_\theta$ is the universal C$^*$-algebra generated by two unitaries $u,v$ subject to the commutation relation $uv=e^{2\pi i \theta} vu$. Setting  $\theta =  \alpha - \beta$, we have a representation $\pi_\theta: A_\theta \rightarrow {\mathcal T}_{\alpha, \beta}$ such that $\pi_\theta(u)=U, \pi_\theta(v)=V$. Thus, the Harper operator 
$H $ is the image of a bounded self-adjoint element of $A_\theta$.

\subsection{Rational vs irrational}
If the parameters  $\alpha, \beta$ are rational numbers, so is $\theta$, then the rotation algebra $A_\theta$ is isomorphic to the continuous sections of some vector bundle over the two-torus $\mathbb{T}^2$.  For rational parameter, the one dimensional Harper operator is a periodic operator and its spectrum consists of energy bands separated by gaps by Bloch-Floquet theory. The Bloch variety associated to a Harper operator with rational phases has only finitely many components, and can be treated similarly to the periodic Schr{\"o}dinger operators discussed in \cite{GKT}. 

When $\theta$ is an irrational real number, the irrational rotation algebra $A_\theta$ is a simple $C^*$-algebra and it has been studied in noncommutative geometry motivated by Kronecker foliation, deformation theory etc. For irrational parameter, the spectrum of the Harper operator is a Cantor set of zero measure \cite{H}, \cite{HS}. 

Here, instead of following the approach to the spectral theory of Harper and almost Mathieu
operators based on functional analysis and noncommutative geometry, we aim at adapting
the algebro-geometric setting developed in \cite{GKT} in the case of discretized periodic
Schr\"odinger operators without magnetic field.

As we show in \S \ref{HarperAGsec} below, in the case of irrational parameters,
the analog of the Bloch variety that describes the complex energy-crystal 
momentum dispersion relation is no longer an ordinary algebraic variety, since
it is defined by a countable family of polynomial equations. It can still be described
in such a way that algebro-geometric methods apply, as a ``limit" of algebraic
varieties, an ind-variety that has infinitely many components. This description 
brings in some additional difficulty when we compactify it and blow it up to 
resolve the singularities. 

This second operation requires one more limit operation, which needs to
be compatible with the first one. An algebro-geometric setting where such
double limit operations can be described and the necessary compatibilities
are encoded in the structure is that of ind-pro-varieties, which can be viewed 
as locally compact Hausdorff totally disconnected spaces, and behave in
all ways more like Cantor sets than algebraic varieties (see \cite{KV}, \cite{P}).
Thus, this different behavior reflects at the level of Bloch varieties
the different structure (bands or Cantor sets) of the Hofstadter butterfly spectrum in the
rational and irrational cases. 

We assume that $\alpha, \beta, \theta$ are all irrational real numbers in this paper, which will not
 be stated otherwise.

\subsection{Harper and almost Mathieu operators}

A limit case of the Harper operator is the almost Mathieu operator, whose spectral theory has been widely studied in connection with the famous phenomenon of the Hofstadter butterfly, which was first observed by Hofstadter\cite{H}, also cf. \cite{B}, \cite{L1}.

Indeed, if we set the parameter $\beta = 0$ and let the Harper operator act on $\ell^2(\mathbb{Z})$, we can express the resulting operator in terms of two new unitaries
$$
U' \varphi(n) : =e^{-2 \pi i \alpha n } \psi(n) \ \  \text{ and } \ \ V'\varphi(n): = \varphi(n+1)
$$ 
Thus the almost Mathieu operator is defined as $H' := U' + U'^* + V' + V'^*$, namely
\begin{equation}\label{almostmathieu}
H' \varphi (n) : =  2cos(2\pi \alpha n)\varphi(n) + \varphi(n+1) + \varphi(n-1)
\end{equation}
The physical meaning of this limit process is that one takes the Landau gauge which forces the vector potential only in one direction, so the almost Mathieu operator is sometimes called the Landau Hamiltonian by physicists. 

Therefore, we get another representation $\pi_{\alpha}$ of the noncommutative torus such that $\pi_{\alpha}(u)= U', \pi_{\alpha}(v)=V' $. In the literature, the almost Mathieu operator is also referred to as the (one dimensional) Harper operator. Since the $C^*$-algebra $A_\theta$ is simple for irrational $\theta$,  the almost Mathieu operator has the same spectrum as that of the (two dimensional) Harper operator. 

The spectral property of the Harper operator is related to the famous ``Ten Martini problem'', which was one of Simon's problems \cite{Sim}. Now its spectrum is well understood and the reader can consult the sources listed as follows. Applying semi-classical analysis based on  Wilkinson's renormalization,  Hellfer and Sj\"ostrand \cite{HS} proved that the spectrum of the almost Mathieu operator has a Cantor structure for frequency $\alpha$ having continued fraction expansions with big denominators. Last \cite{L} also showed that it is a zero measure Cantor set for $\alpha$ satisfying a Diophantine condition.

\begin{center}
\includegraphics[scale=6]{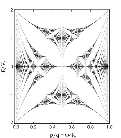}\\
Figure 1: Energy spectrum $ ( -2 \leq E \leq 2) $ as a function of the magnetic flux $\alpha = \Phi/\Phi_0$ where rational numbers $ 0 < p/q < 1$ are used to approximate $\alpha$.
\end{center}

\subsection{Spectral density and parameters}

In the discrete approximation, we study the geometry of a family of Bloch varieties and Fermi curves associated to the spectral theory of the two dimensional and the one dimensional Harper operators, and we compute their density of states. 
For the (two dimensional) Harper operator, we show in \S \ref{DenSec1} that the density of states turns out to be a period independent of the parameter $\theta=  \alpha - \beta$, which involves the complete elliptic integral of the first kind. However, for the almost Mathieu operator (one dimensional case), 
we show by a similar technique in \S \ref{DenSec2} that the density of states explicitly depends on the parameter $\alpha$ in a way that recovers the usual phenomenon of the Hofstadter butterfly, illustrated in the figure. 

Note that, besides the difference in dimension, the generators of the Harper operator have a symmetric form in the variables and parameters, while the almost Mathieu operator is naturally generated by translation and rotation operators. 

Because of the electron-hole symmetry, the density of states of the Harper operator takes the same form for positive and negative energy levels. 
We show in \S \ref{DenSec1} that it can also be represented as a sum of half-periods of two isomorphic elliptic curves. 

When the density of states is available, some spectral functions, such as the zeta function and the partition function, can be obtained by integrations of special functions of the eigenvalues. 

\section{Harper operator and algebraic varieties}\label{HarperAGsec}

In this section, we first describe the geometry of a family of Bloch varieties and algebraic Fermi curves associated to the Harper operator in the discrete model.  Our discussion here follows closely the setting and notations used in \cite{GKT} and we will point out explicitly where and how our case differs from the case without magnetic field considered there. 

While the rational case is essentially like the original case of \cite{GKT}, in the case
of irrational parameters new phenomena arise, which reflect the different structure 
of the spectrum (Cantor sets instead of bands).

We first show that the analog of the Bloch variety, which describes the 
complex energy-crystal momentum dispersion relation, is no longer 
defined by a finite set of polynomial equations but by a countable family
of such equations. Thus, it is no longer an algebraic variety, but we
can still describe it as an inductive limit of algebraic varieties (an ind-variety)
by a procedure that essentially amounts to working with an approximating
family of algebraic varieties obtained by considering only finitely many
components at a time, with compatibility conditions.

We then deal with the problem of the compactification and singularities.
Using similar embedding and blowup techniques as in \S 5 and \S 7 of \cite{GKT}, we describe the compactification of our Bloch ind-variety.  This requires performing the
necessary operations compatibly on the successive approximations of the
ind-variety. The resulting system of operations can be described as a double
limit procedure and it is encoded in a compatible system of embeddings and
projections. This is a special case of a general procedure devised precisely
to the purpose of making sense of such geometric double limits, which is
called an ind-pro-variety, and which in our case gives the algebro-geometric
counterpart of the Cantor sets in the spectrum, in the form of one of these ind-pro-objects
that, although obtained from families of algebraic varieties, behave themselves
like Cantor sets.

We carry out this construction in the subsections \S \ref{HarperAGsec1}--\S \ref{HarperAGsec5}.
We then focus, in the later subsections starting with \S \ref{DenSec1},
on the computation of the density of states and results based on that, such as the spectral functions, using period integrals computed on the spaces constructed in \S \ref{HarperAGsec1}--\S \ref{HarperAGsec5}.

\subsection{A family of Bloch varieties}\label{HarperAGsec1}

As in \cite{GKT}, one defines the Bloch variety associated to the Harper operator \eqref{Harper} as
\begin{equation}\label{BHarp}
\begin{array}{rll} B := & \{ (\xi_1, \xi_2, \lambda)  \in   
\mathbb{C}^* \times \mathbb{C}^* \times \mathbb{C}
 \, | \, H \psi =\lambda \psi, \\  & \psi(m +a,n) =
\xi_1 \psi(m,n), \, \psi(m ,n +b) = \xi_2 \psi(m,n) \}
\end{array}
\end{equation}
This looks very similar to the original case \eqref{BVvar} with trivial potential $V\equiv 0$.
However, in fact, the locus determined by \eqref{BHarp} differs significantly from the case
without magnetic field, since %
in the irrational case it defines a {\em countable} collection of algebraic varieties.

To see this, we first describe the spectral problem defining $B$ 
as a countable collection of matrices acting on the vector $\{\psi(m,n)\} \in \mathbb{C}^{ab}$,
when $\alpha$ and $\beta$ are irrational. This becomes a finite family when both $\alpha$ 
and $\beta$ are rational.

The correct notion of ``limit of algebraic varieties" that we need to employ here to describe
the geometric properties of the resulting space is that of an {\em  ind-variety}. We 
first recall the abstract definition and then we give a more heuristic description of its meaning,
before we apply it to our concrete and specific problem.

\begin{defn} 
An {\em ind-variety } over $\mathbb{C}$ is a set $X$ together with a filtration 
$$
X^0 \subseteq X^1 \subseteq X^2 \subseteq \cdots
$$
such that $\cup_{n \geq 0}X^n = X$ and each $X^n$ is a finite dimensional algebraic variety with the inclusion $ X^n \hookrightarrow X^{n+1}$ being a closed embedding. Such ind-variety $X$ will also be denoted by $ \varinjlim X^n$. An ind-variety $X = \varinjlim X^n $ is said to be {\em affine} (resp. {\em projective}) if each $X^n$ is affine (resp. projective). 
\end{defn}

In categorical terms, an ind-variety $X$ is a formal filtered colimit of an inductive system $ \{ X^n \}$ of varieties. The notion of an ind-variety was first introduced by Shafarevich and one can find more examples and properties of ind-varieties for instance in \cite{K}. 

Let $ X $ and $ Y $ be two ind-varieties with filtrations $ \{ X^n \}$ and $ \{ Y^n \}$. A map $ f: X \rightarrow Y$ is a morphism if and only if for every $ n \geq 0 $, there exists a number $ m(n) \geq 0$ such that  $f|_{X^n}: X^n \rightarrow Y^{m(n)}$ is a morphism between varieties. Thus we get the Ind category Ind$(Var) $ of the category of varieties.

\smallskip
Clearly, the idea here is that one ``builds up" the limit space by a sequence of 
step-by-step operations that progressively enlarge an algebraic variety. The limit
space is no longer a variety itself, but it can be approximated by varieties and 
one can carry over algebro-geometric operations to the limit, as long as they
are performed on the approximating varieties in a way that is compatible with
the successive inclusions that build up the limit space. For example (as will be
directly relevant to our case) the limit may be obtained by adding more and
more components. The abstract definition recalled above is much more
general and allows for more complicated cases where, for example, the
dimensions of the successive approximations grow, so that the limit may also
be infinite dimensional. Our case is milder, as the dimensions of the
components remain bounded and only their number is growing. However,
it is convenient for us to describe the result as an ind-variety as that
encodes all the compatibility conditions between the successive approximations
that we will need to respect when dealing with the compactification and
singularities problem. 
\smallskip

We now proceed to see what all this means concretely, in our specific case.

\begin{lemma}\label{BandMkl}
The Bloch variety \eqref{BHarp} defined by the spectral problem of the Harper
operator \eqref{Harper} is an affine ind-variety that can be written as $B =\cup_{k,\ell \in \mathbb{Z}} B_{k,\ell}$, where
\begin{equation}\label{Bkl}
B_{k,\ell} =  \{ (\xi_1,\xi_2,\lambda)\in {\mathbb C}^*\times {\mathbb C}^* \times
{\mathbb C}  \,|\, P_{k,\ell}(\xi_1,\xi_2,\lambda)=0 \},
\end{equation}
for polynomials
\begin{equation}\label{Pkl}
P_{k,\ell}(\xi_1,\xi_2,\lambda)= \det( M^{(k,\ell)}-\lambda I ),
\end{equation}
where the $ab$ by $ab$ matrix $M^{(k,\ell)}$ has entries
\begin{equation}\label{Mkl}
 \begin{array}{ll}
      0 & \textrm{if $ m'=m,n' = n$}\\
      e^{-2\pi i\alpha (n+\ell b)} & \textrm{if $ m'=m +1,n' = n$}\\
      e^{2 \pi i\alpha (n+\ell b)} & \textrm{if $ m'=m -1,n' = n $}\\
      e^{-2\pi i\beta (m+k a)} & \textrm{if $ m'=m ,n'= n +1$}\\
      e^{2\pi i\beta (m+k a)} & \textrm{if $ m'=m ,n'= n-1$}\\
     - \xi_1 & \textrm{if $ m =1,  m'= a,n' = n $}\\
     - \xi_1^{-1} & \textrm{if $ m =a,  m'= 1,n' = n$}\\
     - \xi_2 & \textrm{if $ m = m',n= 1, n' = b$}\\
     - \xi_2^{-1} & \textrm{if $ m = m',n = b, n' =1$}\\
      0 & \textrm{otherwise}
\end{array}
\end{equation}
\end{lemma}

\proof Consider the spectral problem $H \psi (m,n) =\lambda \psi (m,n)$.
Since we work with the boundary conditions $\psi(m +a,n) =
\xi_1 \psi(m,n)$ and  $\psi(m ,n +b) = \xi_2 \psi(m,n)$, we can consider the
range where $m=1,\ldots, a$ and $n=1,\ldots, b$. If $\alpha$ and $\beta$ are
irrational numbers, then the phase factors $\exp(2\pi i \alpha n)$ and 
$\exp(2\pi i \beta m)$ are not periodic. This means that, for each 
$(m,n)$ in the chosen fundamental domain of the $a{\mathbb Z} \oplus b{\mathbb Z}$ action,
we have a collection of problems, parameterized by the choice of an
element $(k,\ell)$ in ${\mathbb Z}^2$, which differ only in the presence of the phase factors
$\exp(2\pi i \alpha (n+\ell b))$ and $\exp(2\pi i \beta (m+k a))$. In the case where $\alpha$
is rational, the phase factors $\exp(2\pi i \alpha (n+\ell b))$ repeat periodically, with only
finitely many distinct values, and so for $\exp(2\pi i \beta (m+k a))$, when $\beta$
is rational. Thus, in the case where both $\alpha$ and $\beta$ are rational, 
there are only finitely many different varieties $B_{k,\ell}$ to consider, so their union is 
a genuine algebraic variety. While in
the case where at least one of the two parameters is irrational there are infinitely
many components, so their union $B$ is an ind-variety. For each pair $(k,\ell)$ then one can write the corresponding problem in the form given by the matrix \eqref{Mkl}, by arguing as 
in \S 2 of \cite{GKT}. In addition, by the form of the defining matrix, $B_{k,\ell}$ is symmetric under the involution on each fiber, $ (\xi_1, \xi_2, \lambda) \mapsto (\xi_1^{-1}, \xi_2^{-1}, \lambda)$. 

\endproof

The results of this subsection take care of the first of the two geometric
problems mentioned in the introduction. In the following subsections \S 
\ref{HarperAGsec2}--\S \ref{HarperAGsec5} we deal with the second problem:
compactification and singularities. We show that we can resolve it in a way
that is compatible with the limit procedure described in this subsection.
This will be the more technical algebro-geometric part of the paper, and
it can be skipped by the readers who wish to see
directly the derivation of the density of states and spectral functions in
terms of periods, and who are willing to assume, without going through a
more detailed explanation, that we can indeed make 
sense of the double limit procedure needed to obtain the correct geometric 
space on which these period computations take place. 

\subsection{Fourier modes description}\label{HarperAGsec2}

In this subsection we just rewrite our varieties with a convenient change of
coordinates, coming from Fourier transform, which will be useful later, when
we deal with the compactification and singularities problem. 

\smallskip

Let the Fourier transform of $\psi \in \ell^2(\mathbb{Z}^2)$ be
$\tilde{\psi} \in L^2(\mathbb{R}^2/\mathbb{Z}^2)$, namely
$$
\tilde{\psi}(k_1, k_2) = \sum_{(m,n) \in \mathbb{Z}^2} \psi(m,n)
e^{2 \pi i(mk_1 + nk_2)} = \sum_{(m,n) \in \mathbb{Z}^2} a_{mn}z_1^mz_2^n 
$$
where $z_1 :=e^{2 \pi ik_1 }, z_2 :=e^{2 \pi ik_2 } $ and $a_{mn} := \psi(m,n)$.
Furthermore, the Fourier transform of $H\psi$ is
\begin{equation} \label{FTH}
 \begin{array}{ll}
\widetilde{H \psi}(k_1,k_2) =  & e^{-2 \pi ik_1}\tilde{\psi}(k_1,k_2-\alpha) 
                               + e^{2 \pi ik_1}\tilde{\psi}(k_1,k_2+\alpha) \\
                              &+ e^{-2\pi ik_2}\tilde{\psi}(k_1 -\beta ,k_2) 
                               +e^{2 \pi ik_2}\tilde{\psi}(k_1 + \beta,k_2) 
  \end{array}
\end{equation}
i.e.
\begin{equation}\label{Mode}
\widetilde{H \psi}(z_1,z_2)  =  \sum_{(m,n) \in \mathbb{Z}^2} (e^{ -2 \pi i\alpha n} z_1^{-1} + e^{ 2 \pi i\alpha n} z_1    + e^{ -2 \pi i\beta m} z_2^{-1} + e^{ 2 \pi i\beta m} z_2) a_{mn}z_1^mz_2^n 
\end{equation}
Thus the spectrum of $H$ is given by the Fourier modes in the $k$-space:
\begin{equation*}
\lambda = e^{ -2 \pi i\alpha n} z_1^{-1} + e^{ 2 \pi i\alpha n} z_1    + e^{ -2 \pi i\beta m} z_2^{-1} + e^{ 2 \pi i\beta m} z_2 
\end{equation*}
If we consider functions in $\ell^2(\mathbb{Z}^2/a\mathbb{Z}+b\mathbb{Z})$, then by introducing $(k,\ell) \in \mathbb{Z}^2$ 
\begin{equation}\label{Eigenvalue}
\lambda = e^{ -2 \pi i\alpha (n+\ell b)} z_1^{-1} + e^{ 2 \pi i\alpha (n+\ell b)} z_1    + e^{ -2 \pi i\beta (m+ka)} z_2^{-1} + e^{ 2 \pi i\beta (m+ka)} z_2 
\end{equation}

We then proceed, for a fixed $(k,\ell)$, as in \S 2 of \cite{GKT}. We
introduce the unramified covering 
$$
 \begin{array}{l}
   c:\mathbb{C}^* \times \mathbb{C}^* \times  \mathbb{C} \rightarrow \mathbb{C}^* \times \mathbb{C}^*                 \times  \mathbb{C} \\
     ~~~~~~ (z_1, z_2, \lambda) ~~\mapsto ~~~ (z_1^a, z_2^b, \lambda)
  \end{array}
$$
and the preimage
$\tilde{B}_{k,\ell} : = c^{-1}(B_{k,\ell})$. So the
covering $c:\tilde{B}_{k,\ell} \rightarrow B_{k,\ell}$ has the structure group $\mu_a \times \mu_b$, where by $\mu_n$ we mean the group of roots of unity of order $n$, with the
action of $\rho \in \mu_a \times \mu_b$ on the fibers of the form
$\rho \cdot (z_1, z_2, \lambda) = (\rho_1 z_1, \rho_2 z_2, \lambda)$. We write
$\tilde{P} = P \circ c$.

\smallskip

We now show how the discrete Fourier transform description of the Bloch varieties
given in \S 2 of \cite{GKT} is affected by the presence,  in the Harper operator, of the 
phase factors $\exp(2\pi i \alpha)$ and $\exp(2\pi i \beta)$ and their powers. 

\smallskip

For $(m,n) \in \mathbb{Z}^2$, $\{ z_1^mz_2^n  \}$ consists of a basis for the functions in $L^2(\mathbb{T}^2)$. There is an obvious action of  $ \mu_a \times \mu_b$ on this basis $\rho\cdot z_1^mz_2^n :=\rho_1^m\rho_2^nz_1^mz_2^n $, which is basically a change of base, and can be naturally extended to an action on $L^2(\mathbb{T}^2)$. 

Let $\rho$ act on $\widetilde{H \psi}(z_1,z_2)$. It not only changes the basis from $\{ z_1^mz_2^n  \}$ to $\{ \rho_1^m\rho_2^nz_1^mz_2^n \} $, but it also changes the Fourier modes into 
\begin{equation*}\label{NewEv}
e^{ -2 \pi i\alpha(n+\ell b)} \rho_1^{-1} z_1^{-1} + e^{ 2 \pi i\alpha (n+\ell b)}  \rho_1 z_1    + e^{ -2 \pi i\beta (m+ka)}  \rho_2^{-1} z_2^{-1} + e^{ 2 \pi i\beta (m+ka)}  \rho_2 z_2 
\end{equation*}
Fix $(\rho_{01}, \rho_{02}) = (e^{ 2 \pi i/a} , e^{ 2 \pi i /b})$, other roots of $\mu_a$ and  $\mu_b$ can be written   as $ (\rho_1, \rho_2) = (\rho_{01}^p, \rho_{02}^q) $ for some integers $  1 \leq p \leq a  $ and $ 1 \leq q \leq b $, then rewrite the Fourier modes as
\begin{equation*}
 \rho_{01}^{-\alpha(n+\ell b) a -p} z_1^{-1} + \rho_{01}^{\alpha (n+\ell b) a +p} z_1 + \rho_{02}^{-\beta (m+ka) b - q} z_2^{-1} + \rho_{02}^{\beta (m+ka) b +q} z_2
 \end{equation*}

\begin{lemma}\label{Hrho}
The Harper operator \eqref{Harper} determines a family of operators $H^{(k,\ell)}$,
for $(k,\ell)\in {\mathbb Z}^2$, which acts as multiplication by the complex number
\begin{equation}\label{rhoetaz}
 \rho_{01}^{\alpha (n+\ell b) a +p} z_1 + \rho_{01}^{-\alpha(n+\ell b) a -p} z_1^{-1} + \rho_{02}^{\beta (m+ka) b +q} z_2 + \rho_{02}^{-\beta (m+ka) b - q} z_2^{-1} 
\end{equation}
\end{lemma}

\begin{proof}
In addition to the above discussion, we also have to take care of the boundary conditions. Let us look at one of them 
$ \psi(m +a,n) = \xi_1 \psi(m,n)$. Taking the Fourier transform on both sides gives $e^{-2\pi i ak_1} \tilde{\psi}(k_1,k_2) = \xi_1 \tilde{\psi}(k_1,k_2)$, or equivalently $\tilde{\psi}(z_1,z_2) = z_1^a \xi_1 \tilde{\psi}(z_1,z_2)$. 
So the boundary conditions are removed naturally in the covering space $ \tilde{B}_{k,\ell} $ if we set $\xi_1 = z_1^{-a}$ and  $\xi_2 = z_2^{-b}$. In fact, take the involution symmetry $ (\xi_1, \xi_2, \lambda) \mapsto (\xi_1^{-1}, \xi_2^{-1}, \lambda)$ into account, we can just set $\xi_1 = z_1^a$ and  $\xi_2 = z_2^b$. 
\end{proof}

We write $\hat M^{(k,\ell)}$ for the diagonal $ab\times ab$ matrix with entries
\begin{equation}\label{hatMkl}
\hat M^{(k,\ell)}_{m,n}(\rho,z):=  \rho_{01}^{\alpha (n+\ell b) a +p} z_1 + \rho_{01}^{-\alpha(n+\ell b) a -p} z_1^{-1} + \rho_{02}^{\beta (m+ka) b +q} z_2 + \rho_{02}^{-\beta (m+ka) b - q} z_2^{-1} 
\end{equation}
Thus for pairs $ (k,\ell) $,  $\tilde{B}_{k,\ell}$ is the zero-set
of $\tilde{P}_{k,\ell}(z_1,z_2, \lambda) = det(\hat{M}^{(k,\ell)} -
\lambda I)$, 
or equivalently it is given by the zero locus
\begin{equation}
\tilde{B}_{k,\ell} = \{ (z_1,z_2, \lambda)
 \,| \, \prod_{m,n, \rho } (\hat M^{(k,\ell)}_{m,n}(\rho,z) -\lambda)=0 \}.
\end{equation}
For brevity, hereafter we denote $\varrho := ((m,n), (\rho_1, \rho_2)) \in \mathbb{Z}_a \times  \mathbb{Z}_b \times \mu_a \times \mu_b $.

\begin{cor}\label{tildeBkl}
The lifted Bloch ind-variety $\tilde{B} = \bigcup_{(k,\ell) \in \mathbb{Z}^2} \tilde{B}_{k,\ell} = \bigcup_{k,\ell, \varrho} \tilde{B}_{k,\ell, \varrho}$ is a reduced affine ind-variety 
and its components are nonsingular subvarieties
\begin{equation}\label{tildeBklrho}
\tilde{B}_{k,\ell,\varrho} := \{ (z_1, z_2, \lambda) \in \mathbb{C}^* \times \mathbb{C}^* \times  \mathbb{C}~ | ~ \hat M^{(k,\ell)}_{m,n}(\rho,z) -
\lambda = 0 \}
\end{equation}
where $(m,n) \in \mathbb{Z}_a \times  \mathbb{Z}_b$ and $ (\rho_1, \rho_2)\in \mu_a \times \mu_b $.
\end{cor}

Obviously we have a group action of $\mu_a \times \mu_b$
on each $\tilde{B}_{k,\ell,\varrho}$ by acting on each fiber $\rho \cdot (z_1, z_2, \lambda) = (\rho_1 z_1, \rho_2 z_2, \lambda)$, in other words the group action changes the powers $p, q$ in $\hat M^{(k,\ell)}_{m,n}(\rho,z)$ where $\rho = (\rho_1, \rho_2) = (\rho_{01}^p, \rho_{02}^q)$.

\subsection{Singularity locus}\label{HarperAGsec3}

In this subsection we describe explicitly the singularities of the components
of the Bloch ind-varieties that we need to deal with.

For fixed $(k,\ell)$ and $\varrho$, consider a typical fiber $E_\lambda = E_\lambda(k, \ell, \varrho)$ of $\tilde{B}_{k,\ell,\varrho}$ given by the set 
\begin{equation}
 \{ (z_1, z_2) | \lambda = \rho_{01}^{\alpha (n+\ell b) a +p} z_1 + \rho_{01}^{-\alpha(n+\ell b) a -p} z_1^{-1} + \rho_{02}^{\beta (m+ka) b +q} z_2 + \rho_{02}^{-\beta (m+ka) b - q} z_2^{-1} 
\}
\end{equation}
and take the derivatives formally
\begin{equation}
 \begin{array}{ll}
\frac{\partial E_\lambda}{\partial
z_1} = \rho_{01}^{\alpha (n+\ell b) a +p} - \rho_{01}^{-\alpha (n+\ell b) a -p}z_1^{-2}\\
 \frac{\partial E_\lambda}{\partial
z_2} = \rho_{02}^{\beta (m+ka) b +q} - \rho_{02}^{-\beta (m+ka) b - q} z_2^{-2}
 \end{array}
\end{equation}
where we use the same notation $ E_\lambda $ for the variety and for the polynomial
$$
 E_\lambda = \rho_{01}^{\alpha (n+\ell b) a +p} z_1 + \rho_{01}^{-\alpha(n+\ell b) a -p} z_1^{-1} + \rho_{02}^{\beta (m+ka) b +q} z_2 + \rho_{02}^{-\beta (m+ka) b - q} z_2^{-1} - \lambda
$$
Then the singular points  consist of four points $( \pm \rho_{01}^{-\alpha (n+\ell b) a -p} , \pm \rho_{02}^{-\beta (m+ka) b - q})$.

We have an analog of Lemma 5.1 of \cite{GKT}. When $\lambda = 0$, $E_\lambda$ splits into two components
\begin{equation}
 \begin{array}{ll}
 \{ (z_1, z_2) |  \rho_{01}^{\alpha (n+\ell b) a +p} z_1 + \rho_{02}^{\beta (m+ka) b +q} z_2 = 0\}\\
  \{ (z_1, z_2) |\rho_{01}^{\alpha (n+\ell b) a +p} z_1 + \rho_{02}^{-\beta (m+ka) b - q} z_2^{-1}  = 0 \}
  \end{array}
\end{equation}

When $\lambda = 4 $ , $E_\lambda$ is irreducible with 
$( \rho_{01}^{-\alpha (n+\ell b) a -p} , \rho_{02}^{-\beta (m+ka) b - q})$ as the
only singular point. 

When $\lambda = -4$ , $E_\lambda$ is irreducible with 
$( - \rho_{01}^{-\alpha (n+\ell b) a -p} , - \rho_{02}^{-\beta (m+ka) b - q})$ as its singular point. 

Otherwise, when $\lambda \in \mathbb{C} \setminus \{0, \pm 4\}$, the typical fiber is a nonsingular complex curve.

\subsection{Fermi curves}\label{HarperAGsec4}

In this subsection we identify the components of the Fermi curve associated to
the Bloch ind-variety. It is these countably many components of the Fermi
curve that will contribute to the period computation of \S \ref{DenSec1}, after
taking care of the compactification and singularities problem.

Recall that affine Fermi curves are defined by $F_\lambda(\mathbb{C}) :=
\pi^{-1}(\lambda)$ with the projection $\pi: B \rightarrow
\mathbb{C}$; $(\xi_1, \xi_2,\lambda) \mapsto \lambda$. Since 
$B = \tilde{B} / \mu_a \times \mu_b$, $F_\lambda(\mathbb{C})$ is given by 
\begin{equation}\label{FC}
 \bigcup_{k,\ell, m,n}\{ (\xi_1, \xi_2) | \lambda =  e^{ 2 \pi i\alpha (n+\ell b)} \xi_1  + e^{ -2 \pi i\alpha (n+\ell b)} \xi_1^{-1} +  e^{ 2 \pi i\beta (m+ka)} \xi_2 + e^{ -2 \pi i\beta (m+ka)} \xi_2^{-1}  \}
\end{equation}
where $ 1 \leq m \leq a, 1 \leq n \leq b $ and ${(k,\ell)}$ running through $\mathbb{Z}^2$, so we have countably many components $ F_\lambda^{k,\ell,m,n} $ for each $\lambda$. Then the Fermi curve itself is, in this case, an ind-variety. Namely, for fixed integers $ (k, \ell) $ and $ (m,n) $,  the component $F_\lambda^{k,\ell,m,n}$ is given by
\begin{equation}
 \{ (\xi_1, \xi_2) | \lambda =  e^{ 2 \pi i\alpha (n+\ell b)} \xi_1  + e^{ -2 \pi i\alpha (n+\ell b)} \xi_1^{-1} +  e^{ 2 \pi i\beta (m+ka)} \xi_2 + e^{ -2 \pi i\beta (m+ka)} \xi_2^{-1}  \}
\end{equation}
and the singular locus
of $F_\lambda^{k,\ell,m,n}$ is easily derived from that of $E_\lambda$. 

When $\lambda = 0$, $F_\lambda^{k,\ell,m,n}$ has two components
\begin{equation}
 \begin{array}{ll}
 \{ (\xi_1, \xi_2) |  e^{ 2 \pi i\alpha (n+\ell b)} \xi_1 +  e^{ 2 \pi i\beta (m+ka)} \xi_2  = 0 \}\\
  \{ (\xi_1, \xi_2) | e^{ 2 \pi i\alpha (n+\ell b)} \xi_1 + e^{ -2 \pi i\beta (m+ka)} \xi_2^{-1} = 0  \}
 \end{array}
\end{equation}

When $\lambda =  4 $, $F_\lambda^{k,\ell,m,n}$ is singular only at  $
(e^{ -2 \pi i\alpha (n+\ell b)}, e^{ -2 \pi i\beta (m+ka)})$, similarly when $\lambda =  -4 $, $F_\lambda^{k,\ell,m,n}$ is  singular only at  $(- e^{ -2 \pi i\alpha (n+\ell b)}, - e^{ -2 \pi i\beta (m+ka)})$ and they are irreducible curves.

\subsection{Compactification and blowups}\label{HarperAGsec5}

In this more technical subsection we explicitly describe the form of the
compactification of the components of the Bloch ind-variety and the
compatibility of this compactification operation with the limit procedure
described in \S \ref{HarperAGsec1} above. We also describe the
blowup procedures that take care of the singularities problem and
again check that these can be carried out compatibly with the
operation of passing to the limit. Performing both of these operations 
will create a more complicated ``double limit" procedure, which can be
appropriately described, with all the compatibility conditions directly
encoded, by the notion of an ind-pro system of varieties. The resulting
double limit obtained in this way is called an ind-pro-variety and is then
the geometric space that we need to deal with, which describes the
complex energy-crystal momentum dispersion relation in the case
of the Harper and almost Mathieu operators with irrational parameters.

We proceed as in $\S 4$ of \cite{GKT} to obtain the compactifications of the components of our Bloch ind-variety. We use the same notation and terminology as in $\S 4$ of \cite{GKT}. 

By adding points at infinity, first we have $\tilde{B}_1$ as the projective closure  of
$\tilde{B}$ in $\mathbb{P}^1 \times \mathbb{P}^1
\times \mathbb{P}^1$. Let $ s, t $ belong to the set $\{ 0, \infty \}$, it is easy to see that the intersection  $\tilde{B}_1 \bigcap (\mathbb{P}^1 \times \mathbb{P}^1 \times
\mathbb{P}^1 \setminus  \mathbb{C}^* \times
\mathbb{C}^* \times \mathbb{C})$ consists of eight rational curves $  \{ s \}
 \times \mathbb{P}^1 \times \{ \infty \}$, $  \mathbb{P}^1
\times \{ t \} \times \{ \infty \}$, $  \{ s \} \times \{ t \} \times \mathbb{P}^1$ as in Lemma 4.1 of \cite{GKT}.

$\tilde{B}_1$ is singular at the points $O_{s,t}:= ( s, t, \infty) $, after
blowing up these four points, we define $\tilde{B}_2$ as the strict transform of $\tilde{B}_1$ and the group action of  $\mu_a \times \mu_b$ can  be lifted onto  $\tilde{B}_2$ naturally.

Let $P^{s,t}$ be the exceptional divisor over $ O_{s,t}$, the intersection points of $P^{s,t}$  with the strict transforms of eight rational curves are denoted by  $w_0^{s,t} \in P^{s,t} \cap \overline{\{ s \} \times \{ t \} \times  \mathbb{P}^1 
\setminus  O_{s,t}}$,  $w_1^{s,t} \in P^{s,t} \cap \overline{\{ s \} \times   \mathbb{P}^1  \times \{ \infty \} \setminus  O_{s,t}}$ and $w_2^{s,t} \in P^{s,t} \cap \overline{ \mathbb{P}^1  \times \{ t \} \times \{ \infty \} \setminus   O_{s,t}}$. Obviously $w_0^{s,t}$, $w_1^{s,t}$, $w_2^{s,t}$ are fixed points of $\mu_a \times \mu_b$. 

We then have a direct analog of Lemma 4.2 of \cite{GKT} adapted to our inductive system of varieties. The main difference lies in the fact that here we deal with a countable family of quadrics arising from the blowups instead of having just $ ab$ of them as in the original case of \cite{GKT}.

\begin{lemma}\label{lem}
 There exist quadrics
$Q_{k, \ell, \varrho}^{s,t} \subset P^{s,t}$ containing $w_0^{s,t}$, $w_1^{s,t}$, $w_2^{s,t}$ satisfying that there are neighborhoods $U_i$ of $w_i^{s,t}$ such that $U_i \cap \tilde{B}_2$ consists of countably many branches.
\end{lemma}

\begin{proof}
We only consider what happens around the point $ ( 0, 0, \infty)$, other cases can be treated similarly. 

Using the coordinate system $z_1, z_2, \mu= \lambda^{-1} $ in a neighborhood of  $  ( 0, 0, \infty) $, fix $s=t=0$, then $w_0^{00} = (0,0,1), w_1^{00}= (0,1,0), w_2^{00}= (1,0,0)$.

Recall that  the closure of the lifted Bloch variety $\tilde{B}_1$ is the zero-set of the
determinant of the matrices  $\mu z_1 z_2  (M^{(k,\ell)}_{m,n}(\rho,z) - \lambda I)$ with diagonal
entries
\begin{equation}
\rho_{01}^{\alpha (n+\ell b) a +p}\mu z_1^2 z_2 + \rho_{01}^{-\alpha(n+\ell b) a -p}\mu z_2  +  \rho_{02}^{\beta (m+ka) b +q}\mu z_1 z_2 ^2   + \rho_{02}^{-\beta (m+ka) b - q}\mu z_1  -  z_1 z_2 
\end{equation}

Define 
\begin{equation}
 Q_{k,\ell,\varrho}: = \{ (z_1: z_2: \mu) \in P^{00} ~|~ \rho_{01}^{-\alpha(n+\ell b) a -p}\mu z_2 + \rho_{02}^{-\beta (m+ka) b - q}\mu z_1- z_1 z_2  =0 \}
\end{equation}
therefore $P^{00} \cap \tilde{B}_2$ is the union of infinitely many quadrics $Q_{k,\ell,\varrho}$.  In addition,  $ \tilde{B}_2 $ is nonsingular at every point not being a point of intersection of different quadrics. 

Let us look at $w_0^{00}$ first, from the tangent part $ \rho_{01}^{-\alpha(n+\ell b) a -p} z_2 + \rho_{02}^{-\beta (m+ka) b - q}z_1$, there are infinitely many transversal branches labeled by $(k,\ell), (m,n) $ and $(p,q)$ intersecting at $L :=  \{ z_1 =z_2 =0 \}$.

Then consider $w_1^{00}$, since $w_1^{00}= (0,1,0)$ in the coordinates $z_1 , z_2, \mu $,  introduce $\nu \neq 0, u_1,u_2 $ such that
$$
  z_1 = u_1 \nu, ~~ z_2 = \nu,~~\mu = u_2 \nu
$$
And now $\tilde{B}_2$ is given by the
determinant of the infinite matrix with diagonal entries
\begin{equation}
\rho_{01}^{\alpha (n+\ell b) a +p} u_1^2 u_2 \nu^2 + \rho_{01}^{-\alpha(n+\ell b) a -p} u_2  + \rho_{02}^{\beta (m+ka) b +q} u_1 u_2 \nu^2   + \rho_{02}^{-\beta (m+ka) b - q} u_1 u_2  - u_1
\end{equation}

So we get a singular line $L':= \{ u_1 = u_2 = 0\}$ and the tangent cone of
$\tilde{B}_2$ at $L'$ is the union of infinitely many planes $\{ \rho_{01}^{-\alpha(n+\ell b) a -p} u_2 - u_1 = 0\}$. 

Blow up the line $L'$ by defining new coordinates
$$
u_1 = v_1 v_2, ~~ u_2 = v_2
$$
 The strict transform of $\tilde{B}_2$ then is given
by the determinant of the infinite matrix with diagonal entries
\begin{equation}
\rho_{01}^{\alpha (n+\ell b) a +p} v_1^2 v_2^2 \nu^2 +  \rho_{01}^{-\alpha(n+\ell b) a -p}   + \rho_{02}^{\beta (m+ka) b +q} v_1 v_2 \nu^2   + \rho_{02}^{-\beta (m+ka) b - q} v_1 v_2  - v_1 
\end{equation}
In other words,
\begin{equation} 
  \begin{array}{ll}
\rho_{01}^{-\alpha(n+\ell b) a -p} - v_1  + \rho_{02}^{-\beta (m+ka) b - q}\rho_{01}^{-\alpha(n+\ell b) a -p}  v_2 \\
- \rho_{02}^{-\beta (m+ka) b - q}(\rho_{01}^{-\alpha(n+\ell b) a -p} - v_1) v_2 + \text{higher corrections}
  \end{array}
\end{equation}
We get countably many singular lines $L_{\ell, n, p} : = \{ v_1 = \rho_{01}^{-\alpha(n+\ell b) a -p}, v_2 = 0 \}$ and the tangent cone at each $L_{\ell, n,p} $ is $(\rho_{01}^{-\alpha(n+\ell b) a -p} - v_1 ) + \rho_{02}^{-\beta (m+ka) b - q}\rho_{01}^{-\alpha(n+\ell b) a -p}  v_2$, i.e. it has countably many branches as claimed. 

Finally consider $w_2^{00}= (1,0,0)$, analyze similarly to $ w_1^{00}$ and the strict transform $\tilde{B}_2$ is determined by entries
\begin{equation} 
  \begin{array}{ll}
 \rho_{02}^{-\beta (m+ka) b - q}- v_1  + \rho_{01}^{-\alpha(n+\ell b) a -p} \rho_{02}^{-\beta (m+ka) b - q}  v_2 \\
- \rho_{01}^{-\alpha(n+\ell b) a -p}(\rho_{02}^{-\beta (m+ka) b - q} - v_1) v_2 + \text{higher corrections}
  \end{array}
\end{equation}
Countably many singular lines $L_{k, m, q} : = \{ v_1 = \rho_{02}^{-\beta(m+ka) b -q}, v_2 = 0 \}$
are obtained in this case as well. Let $L_{k, \ell,\varrho}$ denote these singular lines $L_{\ell, n, p} $ and $L_{k, m, q}$ together.

\end{proof}

By Bezout's theorem $Q_{k, \ell, \varrho}^{s,t}$ and
$Q_{k', \ell', \varrho'}^{s,t}$ have a point of intersection
$\tilde{d}^{s,t} $ other than $w_0^{s,t}$,
$w_1^{s,t}$, $w_2^{s,t}$ if and only if $(\ell,n,p) \neq (\ell',n',p')$ and
$(k,m,q) \neq (k',m',q')$. If such $\tilde{d}^{s,t} $ does exist, then it can be proved that they are ordinary double points of $\tilde{B}_2$ as in \cite{GKT}.

As in the above lemma \eqref{lem}, we blow up the
strict transforms of the rational curves $ \{ s \} \times
\{ t \} \times \mathbb{P}^1 $, $ \{ s \} \times \mathbb{P}^1\times \{ \infty \} $ and
$  \mathbb{P}^1 \times \{ t \}  \times \{ \infty \}$, and let
$\tilde{B}_3$ be the strict transform of
$\tilde{B}_2$. Then $\tilde{B}_3$ is nonsingular on the strict transforms of the rational curves
$ \{ s \} \times \{ t \}  \times \mathbb{P}^1$, since all branches are transversal to each other in this case. On the other hand,  it is singular
at countably many lines which lie over the strict transforms of $  \{ e \} \times \mathbb{P}^1 \times \{\infty \}$ and respectively of $  \mathbb{P}^1 \times \{ f \} \times \{ \infty \}$. We now take care of these singular lines $L_{k, \ell,\varrho}$. 

\smallskip

In order to do this, we first need to clarify how we are going to proceed at performing
the ``double limit" that will result when we keep adding more and more components
to the Bloch variety (constructing the ind-variety) and at the same time blowing up
all these singular lines, in a compatible way. The double limit procedure we need
is encoded in the notion of an ind-pro-system of algebraic varieties. Again, we first
recall the general abstract definition of this procedure, then we give a more
heuristic interpretation of its meaning, and then we apply it to our concrete case to
obtain explicitly the construction of the compactification and resolution of the Bloch variety
in the irrational case.

\begin{defn}
An ind-pro system of varieties $\{ X_i^n \} $ is a double indexed set of varieties such that for each $i \geq 0$, $\{ X_i^n \} $ is an inductive system while for each $ n \geq 0 $, it is a projective system and every square of the system is Cartesian, i.e. the horizontal maps are injections and the vertical ones are surjections, 
the diagram commutes and
that the top-left corner is the fibered product of the bottom
and right map.
$$
\xymatrix{
\ar @{} [dr] |{}
X^m_i \ar@{>>}[d]^{\pi_{ij}^m}  \ar@{^{(}->}[r] & X^n_i \ar@{>>}[d]^{\pi_{ij}^n} \\
X^m_j \ar@{^{(}->}[r]        & X^n_j
}
$$
\end{defn}

Using categorical dual notions, we define a pro-object of the category Ind$(Var)$ as a formal cofiltered limit of a projective system of ind-varieties. Further, we can construct the category ProInd$(Var)$, or Pro$_{\aleph_0} $Ind$_{\aleph_0}(Var)$ if the index sets are countable.

Given an ind-pro system of varieties $\{ X_i^n \} $, taking the inductive limit and the projective limit gives the ind-pro-variety $ \varprojlim_i \varinjlim_n X^n_i $, which is an object in Pro$_{\aleph_0} $Ind$_{\aleph_0}(Var)$.

\medskip

What this abstract definition is saying in more heuristic terms is that we are building
up a space by performing simultaneously two kinds of operations on a collection of
algebraic varieties. One type of operation is the one we have already seen when
discussing the ind-limit, namely we organize our varieties in a sequence of inclusions
that progressively adds more and more components. The other operation takes care of
the blowups that progressively remove singularities (these form a projective system,
since the smoother blowups map down by projections to the more 
singular varieties they are obtained from).  
The problem here lies in the fact that the required sequence of blowups needs 
to be performed on an infinite number of components that intersect each other, 
and the way to make this possible is by {\em compatibly} carrying out blowups  
on the approximating varieties of the inductive system. Notice that, to 
make sense of this double limit, we only need to use the very simple categorical notions
of inductive and projective limit. In fact, the categorical definition
above is simply encoding the compatibility condition that makes this kind of
double limit operation possible: it is saying that operations we perform
at the level of the approximating algebraic varieties $X^n_i$, as long as they
are done in a way that is compatible with the maps between the varieties that define
the two limit operations, will {\em carry over to the limit}, even though the
space we obtain as the double limit of this sequence of algebraic varieties
is no longer an algebraic variety but a Cantor-like geometry, which cannot
be directly described within classical algebraic geometry. 

\smallskip

We now return to our specific case of the Bloch varieties for the Harper operator in the
irrational case. 

\begin{prop}
In the case of irrational parameters,  the compactification of the lifted Bloch
ind-variety, is an ind-pro-object defined by a chain of iterated blowups on the
ind-variety $\tilde B$.
\end{prop}

\proof
For fixed $(k,\ell), (m,n) $ and $(p,q)$, we can blow up the line $L_{\ell, n, p} $ or $L_{k, m, q} $ as usual. Let us look at $L_{\ell, n, p} $ here. By introducing new coordinates $w_1, w_2$ such that 
\begin{equation}
 v_1 - \rho_{01}^{-\alpha(n+\ell b) a -p} = w_1 w_2, \quad v_2 = w_2 
\end{equation}
the defining equation becomes 
\begin{equation}
  \begin{array} {l}
\rho_{01}^{-\alpha(n+\ell b) a -p} \rho_{02}^{ -\beta (m+ka) b -q} - w_1 +  \rho_{01}^{-\alpha(n+\ell b) a -p} \rho_{02}^{\beta (m+ka) b +q}  \nu^2  + \rho_{02}^{-\beta (m+ka) b - q} w_1 w_2  \\ + \rho_{01}^{-\alpha(n+\ell b) a -p}w_2 \nu^2 +2w_1 w_2^2 \nu^2 + \rho_{02}^{\beta (m+ka) b +q} w_1 w_2 \nu^2  + \rho_{01}^{\alpha (n+\ell b) a +p} w_1^2 w_2^3 \nu^2 
   \end{array}
\end{equation}

First fix $(k, \ell) = (0,0)$ and let $ \varrho = (m,n, p,q) $ vary, we get $2ab$ singular lines $\{ L_{0, n,p},  L_{0, m,q}  \}$, then blow up these $2ab$ lines, define $A_{0}$ as the strict transform of $\tilde{B}_3$. We start over this process again, now let $|k| \leq 1, |\ell| \leq 1 $ and $ \varrho $ vary, blow up these $6ab$ singular lines $\bigcup_{ |k| \leq 1, |\ell| \leq 1} \{L_{k, l,\varrho} \}$, define $A_{1}$ as the strict transform of $\tilde{B}_3$. In general, we have $A_{i}$ being the strict transform of $\tilde{B}_3$ after blowing up $2(2i+1)ab$ singular lines $\bigcup_{ |k| \leq i, |\ell| \leq i} \{L_{k, l,\varrho} \}$ for $|k| \leq i, |\ell| \leq i $ and all $ \varrho $.

Obviously, the projection maps $\pi_{ij}: A_i \rightarrow A_j$ are surjective for any $  i \geq j \geq 0$. In other words, we have a projective system of ind-varieties since each $A_i$ is an ind-variety derived from the blowup $\sigma:  A_i \rightarrow \tilde{B}_3$ and $\tilde{B}_3$ being an ind-variety. Hence, the projective limit $ A = \varprojlim A_i $ is an ind-pro-variety in Pro$_{\aleph_0} $Ind$_{\aleph_0}(Var)$.  We have the following commutative diagram
$$
\xymatrix{
&A \ar[dl]_{\pi_i} \ar[rd]^{\pi_j}\\
A_i \ar[rd]_{\sigma} \ar@{>>}[rr]^{\pi_{ij}} & &A_j \ar[ld]^{\sigma}  \\
&\tilde{B}_3 }
$$
and it is easy to see our construction satisfies the Cartesian squares. 
\endproof

We call $A$ the compactification of the lifted Bloch ind-variety  $ \tilde{B} $ and denote it by $ \bar{\tilde{B}} $ from now on.  

\medskip

What this result shows is that we can {\em compatibly} resolve the compactification and
singularities problem on each individual finite approximation to the Bloch ind-variety, in
a way that allows us to pass to the limit.  Now, as a consequence of this construction, we
find that the space we obtained as the double limit of this compatible family of
algebraic varieties is indeed a Cantor-like space, but one that has a good approximation
by algebraic varieties.  In fact, Kapranov showed 
in \cite{KV} (see also \cite{P}) 
that the category of locally compact Hausdorff totally disconnected spaces can be identified 
with a full subcategory of IndPro$(Set_0 )$, which means that geometric spaces obtained
by the kind of double limit procedure described above look Cantor-like. 
One can find a sketch of the proof in \cite{P}. In particular, this means that 
the compactification of the Bloch variety of the Harper operator with irrational parameters 
is an ind-pro-variety, which is a Cantor-like geometric space, as one might have 
expected by thinking of it as the complex energy-crystal momentum dispersion relation,
in a case where the band structure in the spectrum is replaced by a Cantor set. The
important additional information that the result above gives us is the fact that this
Cantor-like geometry admits a good approximation by algebraic varieties: this will
be useful in \S \ref{DenSec1} and the rest of the paper below, since it will allow us to
reduce the calculation of the density of states for this Cantor-like geometry to a sequence
of terms that can be computed compatibly over the approximating algebraic varieties,
and that can therefore be identified explicitly with period integrals.

\medskip

We now rephrase the previous result on the Bloch varieties in terms of Fermi curves,
since these will be the curves over which the period calculation of \S \ref{DenSec1}
for the density of states will take place.

We have the analog of Theorem 4.2 of \cite{GKT}. The map $\pi \circ c: \tilde{B}
\rightarrow \mathbb{C} $ extends to a morphism $\tilde{\pi}:
\bar{\tilde{B}} \rightarrow \mathbb{P}^1$ and its fibers
$\tilde{F}_\lambda := \tilde{\pi}^{-1}(\lambda)$ are called the lifted Fermi
curves.  The inclusion map $i: \tilde{B} \hookrightarrow
 \mathbb{C}^* \times \mathbb{C}^* \times \mathbb{C}$ gives rise to a morphism
$\tilde{i}: \bar{\tilde{B}} \rightarrow
\mathbb{P}^1 \times \mathbb{P}^1 \times \mathbb{P}^1$.

\begin{thm}\label{FermitildeB}
 $\bar{\tilde{B}}
\setminus \tilde{B}$ consists of countable curves\\
(i)  sections of $\tilde{\pi}$, $\Sigma_{k,\ell,\varrho}^{s,t}$  with $\tilde{i}(\Sigma_{k,\ell,\varrho}^{s,t}) =  \{ s \}
\times \{ t \}  \times \mathbb{P}^1$ \\
(ii) quadrics $Q_{k,\ell,\varrho}^{s,t}$ with
$\tilde{i}(Q_{k,\ell,\varrho}^{s,t}) = \{ s \} \times \{
t \}  \times  \{ \infty \}$\\
(iii)  $H_{\ell,n,p}^s$ with
$\tilde{i}(H_{\ell,n,p}^s) =   \{ s \} \times
\mathbb{P}^1 \times \{ \infty \}$
and  $H_{k,m,q}^t$ with $\tilde{i}(H_{k,m,q}^t) = \mathbb{P}^1 \times
\{ t \}  \times \{ \infty \} $

These curves meet transversally at only one point, except that $\{ \tilde{d}^{s,t} \}$ are ordinary double points of $\bar{\tilde{B}}$, every intersection point is a nonsingular point of the compactification $\bar{\tilde{B}}$.
\end{thm}

Taking the quotient by the structure group $ \mu_a \times \mu_b $, we get the compactification $\bar{B}  = \bar{\tilde{B}} / \mu_a \times \mu_b $ of the Bloch ind-variety $B$. At the same time, we get two morphisms, the inclusion $i : \bar{B} \rightarrow \mathbb{P}^1 \times \mathbb{P}^1 \times \mathbb{P}^1 $ and the projection $\pi: \bar{B} \rightarrow \mathbb{P}^1$, whose fibers are the so-call compactified Fermi curves $F_\lambda$. Let $\Sigma_{k,\ell,m,n}^{s,t}$, $Q_{k,\ell,m,n}^{s,t}$, $H_{\ell,n}^s$ and  $H_{k,m}^t$ in $\bar{B}$ be the image under the
quotient map, in addition, let $d^{s,t}$ be the image of $\tilde{d}^{s,t}$ in $\bar{B}$.

\begin{thm}
$\bar{B} \setminus B$ is the union
of countable curves,\\
(i)  $\Sigma_{k,\ell,m,n}^{s,t}$ with $i(\Sigma_{k,\ell,m,n}^{s,t}) =
 \{ s \} \times \{ t \}  \times \mathbb{P}^1$, nonsingular points of $\bar{B}$\\
(ii)  $Q_{k,\ell,m,n}^{s,t}$ with ordinary double
points at $d^{s,t}$ and nonsingular at all points of $Q^{s,t}_{k,\ell,m,n} \setminus \bigcup \{ d^{s,t} \}$ \\
(iii)  $H_{\ell,n}^s$ with $i(H_{\ell,n}^s) = \{ s \} \times \mathbb{P}^1 \times \{ \infty \} $ and  $H_{k,m}^t$  with $i(H_{k,m}^t) =  \mathbb{P}^1 \times \{ t \} \times \{\infty \} $, nonsingular points of $\bar{B}$. 
\end{thm}

\subsection{Density of states and periods}\label{DenSec1}

In the above subsections we have constructed the desired geometric space, the compactification of the Bloch ind-variety, which can be approximated by finite dimensional algebraic varieties. In this subsection, we will proceed to compute the density of states on those approximating components for the Harper operator. The density of states is indeed the period integral over Fermi curves, which also related to the periods of elliptic curves. %

In this subsection, we assume $|\xi_1|=|\xi_2| = 1$, i.e. $(\xi_1, \xi_2) = (e^{2\pi i k_1}, e^{2\pi i k_2}) $ for some $ (k_1,k_2) \in (0,1]^2$,  then the Bloch ind-variety is given by
\begin{equation}
  \begin{array}{rl}
B = \{ (e^{2\pi i k_1}, e^{2\pi i k_2}, \lambda)
 ~|~  & H \psi =\lambda \psi,  \\
      & \psi(m +a,n) = e^{2 \pi i k_1} \psi(m,n), \\
      & \psi(m ,n +b) = e^{2\pi i k_2} \psi(m,n) \}
 \end{array}  
\end{equation}
Denote its spectrum by $ \sigma(H) := \{ E_j(k_1, k_2 ), j \in \mathbb{N} \}$, the function $ E_j(k_1, k_2 )$ is the so-called  $j$-th band function. 

In order to compute the density of states, we get back to a continuous model by taking the limit of the lattice model. Here again we use the same notation and terminology as $\S3$ and $\S11$ of \cite{GKT}.

Let $H_n$ denote the Harper operator $H$ acting on $\ell^2(\mathbb{Z}^2/an\mathbb{Z} \oplus bn \mathbb{Z} )$ for some integer $n \geq 1$, later we will take the limit as $n$ tends to infinity. It is easy to see that the
eigenvalues of $H_n$ are just given by
$$
\{ E_j(\frac{m_1}{n} , \frac{m_2}{n} ) |  1 \leq m_1, m_2 \leq n, ~ j \geq 1 ~\}
$$

To define the integrated density of states, we first count the number of eigenvalues less than or equal to $\lambda$. Or equivalently, with the step function $ \Theta(x) $, define
\begin{equation}
\nu_n(\lambda) := \sum_{j=1}^{\infty} \sum_{m_1,m_2 = 1}^n \Theta(\lambda - E_j(\frac{m_1}{n} , \frac{m_2}{n} ))
\end{equation}

The integrated density of states is then defined as the limit 
\begin{equation}
\rho(\lambda)  : =  \lim_{n \to \infty}\frac{1}{abn^2} \nu_n(\lambda) 
\end{equation}
and the density of states is defined as the derivative $d \rho / d \lambda $. In the literature, the integrated density of states can also be defined as the normalized trace on a $II_1$ factor in von Neumann algebra theory \cite{Sh}.

Rewrite it as an integration, we get an integral over the continuous variable $ k \in [0, 1]^2$
\begin{equation}
 \begin{array}{l @{\quad} l}
\rho(\lambda)  &  =  \lim_{n \to \infty}\frac{1}{ab} \sum_{j=1}^{\infty} \frac{1}{n^2}\sum_{m_1,m_2 = 1}^n \Theta(\lambda - E_j(\frac{m_1}{n} , \frac{m_2}{n} ))\\
              &  =  \frac{1}{ ab} \sum_{j=1}^{\infty} \int_{I^2} \Theta(\lambda - E_j(k))dk,  
 \end{array}
\end{equation}
\begin{lemma}
The density of states can be expressed as $\frac{ d \rho }{ d \lambda } = \frac{1}{ ab}  \int_{\lambda \in \sigma(H)} \omega_\lambda$, where $\omega_\lambda $ is a differential $1$-form.
\end{lemma}
\begin{proof}

\begin{equation*}
  \begin{array}{l @{\quad} l}
\frac{ d \rho  }{ d \lambda }&  =  \frac{1}{ ab} \Sigma_{j=1}^{\infty} \int_{ I^2} \delta(\lambda - E_j(k) ) dk\\
                                &  =  \frac{1}{ ab} \Sigma_{j=1}^{\infty} \int_{ E_j(k) = \lambda } \frac{ ds }{ | \nabla_k E_j | }\\
                                 &  =  \frac{1}{ ab} \Sigma_{j=1}^{\infty} \int_{ E_j(k) = \lambda } \frac{\sqrt{dk_1^2 + dk_2^2 }}
                                 {\sqrt{\partial_1 E_j^2 + \partial_2 E_j^2 }}
  \end{array}
\end{equation*}

For fixed $\lambda $, from $E_j(k_1,k_2) = \lambda$, we have $dE_j = \frac{\partial E_j }{\partial k_1}dk_1 + \frac{\partial E_j }{\partial k_2} dk_2 = 0$, 
\begin{equation*}
dk_1 = -\frac{\partial_2E_j}{\partial_1E_j} dk_2 \quad \mbox{or} \quad dk_2 = -\frac{\partial_1E_j}{\partial_2E_j} dk_1
\end{equation*}
Then 
\begin{equation*}
ds = \sqrt{\partial_1 E_j^2 + \partial_2 E_j^2}~ \frac{dk_1}{|\partial_2 E_j|} \quad \mbox{or} \quad  ds = \sqrt{\partial_1 E_j^2 + \partial_2 E_j^2} ~ \frac{dk_2}{|\partial_1 E_j|}
\end{equation*}
Therefore we can write the density of states as
$$
\frac{ d \rho }{ d \lambda } 
      = \frac{1}{ ab} \sum_{j=1}^{\infty} \int_{ E_j(k) = \lambda } \frac{dk_2}{ |\partial_1 E_j| }
      = \frac{1}{ ab} \sum_{j=1}^{\infty} \int_{ E_j(k) = \lambda }  \frac{dk_1}{ |\partial_2 E_j| }
      = \frac{1}{ ab} \int_{ \lambda \in \cup_j \{ E_j(k) \}}  \frac{dk_1}{ |\partial_2 E_j| }
$$
Denote the differential $1$-form by $\omega_\lambda$, i.e. $\omega_\lambda(k_1,k_2) : = \frac{dk_2}{ |\partial_1 E_j| } = \frac{dk_1}{ |\partial_2 E_j| } $, the lemma is proved.
\end{proof}

Let $P(\xi_1, \xi_2, \lambda )$ be a general polynomial with $ \xi_1 = e^{2\pi ik_1} $ and
$ \xi_2 = e^{2\pi ik_2} $. Then $ dk_1 = d \xi_1/2\pi i\xi_1  $ and $ dk_2 = d \xi_2/2\pi i\xi_2  $, we also have
$d\lambda = \partial_1E_j dk_1 + \partial_2E_j dk_2 $ from $\lambda = E_j(k_1,k_2)$. Plug into $dP = P_{\xi_1}d\xi_1 + P_{\xi_2}d\xi_2 + P_\lambda d\lambda = 0$, we get
\begin{equation}
(2\pi i P_{\xi_1}\xi_1 + P_\lambda \partial_1E_j)~ dk_1  +  (2\pi i P_{\xi_2}\xi_2 + P_\lambda \partial_2E_j)~ dk_2 = 0
\end{equation}
Since $dk_1$ and $dk_2$ are independent, so
\begin{equation}
 \left\{ \begin{array}{l @{\quad} l}
2\pi i P_{\xi_1}\xi_1 + P_\lambda \partial_1E_j = 0 \\
2\pi i P_{\xi_2}\xi_2 + P_\lambda \partial_2E_j = 0
\end{array} \right.
\end{equation}

\begin{lemma}
The density of states is a period integral over Fermi curves.
\end{lemma}
\begin{proof}
Recall that the Bloch ind-variety $B = \cup_{(k,\ell)\in \mathbb{Z}^2}\cup_{m,n=1}^{a,b}B_{m,n}^{(k,\ell)}$, where
\begin{equation}
B_{m,n}^{(k,\ell)} = \{ e^{2 \pi i \alpha (n+\ell b)}\xi_1 +  e^{-2 \pi i \alpha (n+\ell b)}\xi_1^{-1} + e^{2 \pi i \beta (m+ka)}\xi_2 +  e^{-2 \pi i \beta (m+ka)}\xi_2^{-1} - \lambda = 0 \}
\end{equation}
and Fermi curves $F_\lambda = \cup_{(k,\ell)\in \mathbb{Z}^2}\cup_{m,n=1}^{a,b}F_\lambda^{k,\ell,m,n}$, where
\begin{equation}
F_\lambda^{k,\ell,m,n} = \{ e^{2 \pi i \alpha (n+\ell b)}\xi_1 +  e^{-2 \pi i \alpha (n+\ell b)}\xi_1^{-1} + e^{2 \pi i \beta (m+ka)}\xi_2 +  e^{-2 \pi i \beta (m+ka)}\xi_2^{-1} = \lambda \}
\end{equation}
$P^{k,\ell,m,n} :=  e^{2 \pi i \alpha (n+\ell b)}\xi_1 +  e^{-2 \pi i \alpha (n+\ell b)}\xi_1^{-1} + e^{2 \pi i \beta (m+ka)}\xi_2 +  e^{-2 \pi i \beta (m+ka)}\xi_2^{-1} - \lambda  $ denote the polynomial, then
\begin{equation}
 \left\{ \begin{array}{l @{\quad} l}
P^{k,\ell,m,n}_\lambda  = -1 \\
P^{k,\ell,m,n}_{\xi_1} = e^{2 \pi i \alpha (n+\ell b)} -  e^{-2 \pi i \alpha (n+\ell b)}\xi_1^{-2} \\
P^{k,\ell,m,n}_{\xi_2} = e^{2 \pi i \beta (m+ka)} -  e^{-2 \pi i \beta (m+ka)}\xi_2^{-2}
\end{array} \right.
\end{equation}
It follows that on each component $F_\lambda^{k,\ell,m,n}$
\begin{equation}
 \left\{ \begin{array}{l @{\quad} l}
\partial_1E_j  = 2\pi i P^{k,\ell,m,n}_{\xi_1} \xi_1 \\
\partial_2E_j  = 2\pi i P^{k,\ell,m,n}_{\xi_2} \xi_2
\end{array} \right.
\end{equation}
Furthermore
\begin{equation}
 \left\{ \begin{array}{ll}
\frac{dk_2}{\partial_1 E_j}  =  \frac{d \xi_2/ 2\pi i \xi_2}{2\pi i P^{k,\ell,m,n}_{\xi_1} \xi_1} =  \frac{1}{(2\pi i)^2} \frac{d \xi_2 }{ \xi_1 \xi_2 P^{k,\ell,m,n}_{\xi_1}} \\
 \frac{dk_1}{ \partial_2 E_j } = \frac{ d \xi_1/2\pi i\xi_1}{ 2\pi i P^{k,\ell,m,n}_{\xi_2} \xi_2 } = \frac{1}{(2\pi i)^2} \frac{d \xi_1 }{ \xi_1 \xi_2 P^{k,\ell,m,n}_{\xi_2}}
\end{array} \right.
\end{equation}
i.e. $\omega_\lambda (\xi_1,\xi_2)= \frac{d \xi_1 }{ 4 \pi^2  |\xi_1 \xi_2 P^{k,\ell,m,n}_{\xi_2} |} =  \frac{d \xi_2 }{4 \pi^2  |\xi_1 \xi_2 P^{k,\ell,m,n}_{\xi_1}|} $. 
By the change of variables $ \xi_1 = e^{2\pi ik_1}, \xi_2 = e^{2\pi ik_2}$, the graph $(k_1,k_2, E_j(k_1,k_2))$ is changed into $(\xi_1, \xi_2, \lambda)$. Then
$$
\frac{ d \rho }{ d \lambda }  =  \frac{1}{ ab} \sum_{j=1}^{\infty} \int_{ E_j(k) = \lambda } \omega_\lambda(k_1,k_2)\\
                              =  \frac{1}{ ab} \sum_{(k, \ell) \in \mathbb{Z}^2}\sum_{m,n=1}^{a,b} \int_{F_\lambda^{k,\ell,m,n}} \omega_\lambda(\xi_1,\xi_2)\\
                              =  \frac{1}{ ab} \int_{F_\lambda} \omega_\lambda
$$
\end{proof}

Notice how in this result we used explicitly the fact that the geometric space describing
the complex energy-crystal momentum relation for the Harper operator with irrational
parameters admits a good approximation by a family of algebraic varieties, as we proved
in \S \ref{HarperAGsec1}--\S \ref{HarperAGsec5}. In particular, we see here that the density
of states is computed as a period integral with compatible contributions from 
each of the components $F_\lambda^{k,\ell,m,n}$ of the Fermi curve, as described in
Theorem \ref{FermitildeB} above.
\medskip

Recall then that the incomplete elliptic integral of the first kind is defined as
\begin{equation}
F(x,k) : = \int_0^x  \frac{dt} {\sqrt{(1  - t^2 )( 1 - k^2t^2)} }
\end{equation}
and the complete elliptic integral of the first kind is defined as
\begin{equation}
K(k) : = F(1,k) = \int_0^1   \frac{dt} {\sqrt{(1  - t^2 )( 1 - k^2t^2)} }
\end{equation}
The complete elliptic integral of the first kind is also called the quarter period.

We rewrite $\omega_\lambda$  on $F_\lambda^{k,\ell,m,n} \cap \{|\xi_1| = |\xi_2|=1 \} $ as
\begin{equation}
\omega_\lambda  =  \frac{ d\xi_1 }{ 4 \pi^2 |\xi_1 \xi_2 (e^{2 \pi i \beta (m+ka)} -  e^{-2 \pi i \beta (m+ka)}\xi_2^{-2} )| }
               = \frac{ dk_1}{4 \pi |sin2\pi (k_2+\beta (m+ka))|}
\end{equation}
Let $C_{m,n}^{(k,\ell)}: = \{ (k_1, k_2) | cos2\pi (k_1+\alpha (n+\ell b)) + cos2\pi (k_2+\beta (m+ka)) = \lambda/2 \} $ so that
\begin{equation}
\frac{ d \rho }{ d \lambda }   
=  \frac{1}{4 \pi ab} \sum_{k,\ell=-\infty}^\infty \sum_{m,n=1}^{a,b}\int_{C_{m,n}^{(k,\ell)}} \frac{ dk_1}{|sin2\pi (k_2+\beta (m+ka))|} 
\end{equation}

\begin{thm}
If we define the elliptic modulus as $k := \frac{4-|\lambda|}{4+|\lambda|}$, then the density of states as a function of $k$ is
$\frac{ d \rho }{ d \lambda }|_{F_\lambda^{k,\ell,m,n}} =  \frac{1}{2\pi^2 ab}(1+k) K(k) =  \frac{1 }{2\pi^2 ab} K(\frac{2\sqrt{k}}{1+k})$ on each component of the Fermi curve.
\end{thm}
\begin{proof}

Denote the definite integral $I_{m,n}^{(k,\ell)}(\lambda) :  = \int_{C_{m,n}^{(k,\ell)}} \frac{ dk_1}{|sin2\pi (k_2+\beta (m+ka))|}$ for $k_1 \in I$ moving on the curve $C_{m,n}^{(k,\ell)}$. 

The differential form $\omega_\lambda$ is well defined, and $k_1 \neq -\alpha (n+\ell b), -\alpha (n+\ell b) \pm 1/2~ (mod ~2\pi)$, we can just assume $ 0 < -\alpha (n+\ell b) < 1/2$. 

Let $k_1$ run through a half period of the sine function, i.e. $ 2\pi (k_1+\alpha (n+\ell b)) \in (0, \pi)$, then $I_{m,n}^{(k,\ell)}(\lambda)$ is twice of this integral over a half period. 
\begin{equation*}
\begin{array}{rl}
  I_{m,n}^{(k,\ell)}(\lambda)  &  = 2 \int_0^{1/2} \frac{ d(k_1+\alpha (n+\ell b))}{\sqrt{1 - ( \lambda/2 - cos2\pi (k_1+\alpha (n+\ell b)))^2}} \\
          &  = 2 \int_0^{1/2} \frac{ dcos2\pi (k_1+\alpha (n+\ell b))}{\sqrt{1 - ( \lambda/2 - cos2\pi (k_1+\alpha (n+\ell b)))^2}( -2\pi sin2\pi (k_1+\alpha (n+\ell b) )) }\\
&  = \frac{-1}{\pi} \int_0^{1/2} \frac{dcos2\pi(k_1+\alpha (n+\ell b))} { \sqrt{1 - ( \lambda/2 - cos2\pi (k_1+\alpha (n+\ell b)))^2} \sqrt{ 1 - ( cos2\pi(k_1+\alpha (n+\ell b)) )^2 } } \\
 &  = \frac{1}{\pi} \int^1_{-1} \frac{dx_1} { \sqrt{1 - ( \lambda/2 - x_1)^2} \sqrt{ 1 - x_1^2 } }
\end{array}
\end{equation*}
in the last line we set $x_1 : = cos2\pi(k_1+\alpha (n+\ell b))$, here $-1 \leq x_1 \leq 1$ and  $-1 \leq x_1 - \lambda/2 \leq 1$, there are two cases:  

$ ~(i) ~\lambda/2 -1 \leq x_1  \leq 1$ when $\lambda \in ( 0, 4 )$; 

$~(ii)~ -1 \leq x_1  \leq 1 + \lambda/2$ when $\lambda \in ( -4, 0 )$.

In the first case, $\lambda \in ( 0, 4) $
\begin{equation*}
\begin{array}{rl}
I_{m,n}^{(k,\ell)}(\lambda)  &  = \frac{1}{\pi} \int^1_{\lambda/2 -1} \frac{dx_1} { \sqrt{1 - ( \lambda/2 - x_1)^2} \sqrt{ 1 - x_1^2 } } \\
 ~~x : = x_1 - \lambda/4 &  = \frac{1}{\pi} \int^{1 - \lambda/4}_{  \lambda/4 -1}  \frac{dx} { \sqrt{1 - ( \lambda/4 - x)^2} \sqrt{ 1 - (x + \lambda/4 )^2 } } \\
&  = \frac{1}{\pi} \int^{1 - \lambda/4}_{  \lambda/4 - 1}  \frac{dx} { \sqrt{((1 - \lambda/4)^2 - x^2 )((1 + \lambda/4)^2 - x^2)} } \\
&  = \frac{1}{\pi} \int^{1 - \lambda/4}_{  \lambda/4 -1}  \frac{dx} { (1 - \lambda/4)(1 + \lambda/4) \sqrt{(1  - x^2/(1 - \lambda/4)^2 )( 1 - x^2/(1 + \lambda/4)^2)} } \\
~~ t : = x/(1 - \lambda/4) &  = \frac{4}{(4+\lambda )\pi} \int^1_{- 1 } \frac{dt} {\sqrt{(1 - t^2 )( 1 - k^2t^2)} } \\
~~ k := ( 4 - \lambda )  / ( 4 + \lambda ) &  = \frac{1+k}{\pi}  K(k)  
\end{array}
\end{equation*}

In the second case, $\lambda \in (-4, 0 ) $
\begin{equation*}
\begin{array}{rl }
I_{m,n}^{(k,\ell)} (\lambda) 
&  = \frac{1}{\pi} \int^{1+ \lambda/2}_{-1} \frac{dx_1} { \sqrt{1 - ( \lambda/2 - x_1)^2} \sqrt{ 1 - x_1^2 } } \\
~~x : = x_1 - \lambda/4 &  = \frac{1}{\pi} \int^{1 + \lambda/4}_{ -1- \lambda/4 }  \frac{dx} { \sqrt{1 - ( \lambda/4 - x)^2} \sqrt{ 1 - (x + \lambda/4 )^2 } } \\
~~t : = x/(1 + \lambda/4) &  =\frac{4}{(4-\lambda )\pi} \int^1_{ -1}  \frac{dt} { \sqrt{(1  - t^2 )( 1 - k'^2t^2)} } \\
k' := ( 4 + \lambda )  / ( 4 - \lambda ) &  = \frac{1+k'}{\pi}  K(k')
\end{array}
\end{equation*} 

By identifying $k=k'$, namely, defining $k := \frac{4-|\lambda|}{4+|\lambda|}$,  we conclude that $I_{m,n}^{(k,l)}(\lambda) = \frac{1+ k }{\pi} K(k)$ and the density of states on each component $F_\lambda^{k,\ell,m,n}$ turns out to be
$$
 \frac{1+k }{2\pi^2 ab} K(k) =  \frac{1 }{2\pi^2 ab}  K(\frac{2\sqrt{k}}{1+k})
$$
\end{proof}

Interchanging positive and negative $\lambda$ corresponds to the electron-hole symmetry in solid state physics. We give a representation of the density of states in terms of elliptic curves.

\begin{cor}
On each component of the Fermi curve the density of states is a sum of two half-periods of isomorphic elliptic curves up to a constant, namely
$\frac{ d \rho }{ d \lambda }|_{F_\lambda^{k,\ell,m,n}}  =  \frac{1}{8\pi^2 ab} (\int_\gamma \omega_k + \int_{\gamma'} \omega_{1/k})$. 
\end{cor}
\begin{proof}
First we rewrite the density of states on each component as 
\begin{equation*}
\begin{array}{rl}
\frac{ d \rho }{ d \lambda }|_{F_\lambda^{k,\ell,m,n}}  & = \frac{1}{4\pi^2 ab} \int_{-1}^1 \frac{dt} {\sqrt{(1  - t^2 )( 1 - k^2t^2)} } +
                \frac{k}{4\pi^2 ab} \int_{-1}^1 \frac{dt} {\sqrt{(1  - t^2 )( 1 - k^2t^2)} }\\
u := kt      & = \frac{1}{4\pi^2 ab} \int_{-k}^k \frac{du} {\sqrt{(k^2  - u^2 )( 1 - u^2)} } +
                \frac{1}{4\pi^2 ab} \int_{-1}^1 \frac{dt} {\sqrt{(1  - t^2 )( 1 - k^2t^2)/k^2} }\\ 
            & = \frac{1}{4\pi^2 ab} \int_{-k}^k \frac{dt} {\sqrt{( t^2 - 1)(t^2  - k^2 )} } +
                \frac{1}{4\pi^2 ab} \int_{-1}^1 \frac{dt} {\sqrt{(t^2 - 1 )( t^2 -1/k^2)} }\\       
\end{array}
\end{equation*}

For the elliptic modulus $k \in (0,1)$, consider the elliptic curves given by
$$
\begin{array}{l @{\quad} l}
E_k~~ := \{ (t,y) | y^2 = ( t^2 - 1)(t^2  - k^2 ) \}\\
E_{1/k} := \{ (t,y) | y^2 = ( t^2 - 1)(t^2  - 1/k^2 ) \} 
\end{array}
$$
If we denote the canonical holomorphic form $\omega= dt/y$ by $\omega_k$ for $E_k$ (resp. $\omega_{1/k}$ for $E_{1/k}$), then
$$
\frac{d\rho }{d\lambda }|_{F_\lambda^{k,\ell,m,n}}  = \frac{1}{4\pi^2 ab} (\int_{-k}^k \omega_k+ \int_{-1}^1 \omega_{1/k})
$$

Recall that for $E_k$, we cut the Riemann sphere from $k$ to $1$ and from $-1$ to $-k$, then assemble a Riemann surface homeomorphic to the torus $\mathbb{T}^2$. Similarly for  $E_{1/k}$, we cut the Riemann sphere from $1$ to $1/k$ and from $-1/k$ to $-1$, also assemble a Riemann surface homeomorphic to the torus $\mathbb{T}^2$. 

A framed elliptic curve $(E, \delta, \gamma)$ is an elliptic curve with an integral basis for the first homology such that the intersection number $\delta \cdot \gamma = 1$. And for the period vector $(\int_\delta \omega, \int_\gamma \omega)$, the ratio $\tau(E,\delta, \gamma) = \int_\delta \omega / \int_\gamma \omega $ is an invariant for isomorphic complex tori, which is called the modulus of the isomorphism class.

On $E_k$, we can choose $\delta$ as the cycle from $k$ to $1$ then from $1$ back to $k$ and $\gamma$ from $-k$ to $k$ then back to $-k$, which makes $(E_k, \delta, \gamma)$ a framed elliptic curve. Similarly for $E_{1/k}$, we choose $\delta'$ as the cycle from $1$ to $1/k$ then back to $1$ and $\gamma'$ from $-1$ to $1$ then back to $-1$.

By deforming the path of integration, $\int_\delta \omega_k = \int_k^1 \omega_k + \int^k_1 \omega_k = 2 \int_k^1 \omega_k$ and $\int_\gamma \omega_k = \int_{-k}^k \omega_k + \int_k^{-k} \omega_k = 2 \int_{-k}^k \omega_k$. 
Then 
$$
\begin{array}{l @{\quad} l}
\tau(E_k) = \frac{\int_\delta \omega_k}{\int_\gamma \omega_k}  = \frac{\int_k^1 \omega_k}{\int_{-k}^k \omega_k}\\
\tau(E_\frac{1}{k}) = \frac{\int_{\delta'} \omega_{1/k}}{\int_{\gamma'} \omega_{1/k}}  = \frac{\int_1^{1/k} \omega_{1/k}}{\int_{-1}^1 \omega_{1/k}}
\end{array}
$$

It is easy to see that $\tau(E_k) = \tau(E_{1/k})$, indeed $ \int_{-1}^1 \omega_{1/k} = k \int_{-k}^k \omega_k$ and $ \int_1^\frac{1}{k} \omega_{1/k} = k \int_k^1 \omega_k $. In other words, we have elliptic curves $E_k$ isomorphic to $E_{1/k}$ and we can interchange between them by rescaling the fundamental domains. 

Hence on each component of the Fermi curve
$$
\frac{ d \rho }{ d \lambda }|_{F_\lambda^{k,\ell,m,n}} 
            =  \frac{1}{8\pi^2ab} (\int_\gamma \omega_k + \int_{\gamma'} \omega_{1/k})    
            =  \frac{1}{8\pi^2\tau ab} (\int_\delta \omega_k + \int_{\delta'} \omega_{1/k})
$$
\end{proof}

\subsection{Spectral functions}\label{SpFunc1sec}

With the density of states in hand, we can continue to derive some interesting spectral functions. In this subsection, we will calculate the partition function of the propagating electron in magnetic field on the components of the Bloch ind-variety.

Fix $(k,\ell,m,n)$, we consider one irreducible component $B_{m,n}^{(k,\ell)}$, other components can be treated similarly.
Recall that 
\begin{equation}
B_{m,n}^{(k,\ell)} = \{e^{2\pi i \alpha (n+\ell b) }\xi_1 + e^{-2\pi i\alpha (n+\ell b) } \xi_1^{-1} + e^{2\pi i \beta  (m+ka)}\xi_2 + e^{-2\pi i \beta (m+ka) } \xi_2^{-1} - \lambda = 0 \}
\end{equation}
For the moment, we omit the superscripts and denote the above defining polynomial as $P( \xi_1, \xi_2, \lambda)$.

From $dP = P_{\xi_1}d\xi_1 + P_{\xi_2}d\xi_2 + P_\lambda d\lambda = 0$, then on $B_{m,n}^{(k,\ell)}$ the pullback 
$$
\pi^*(d\lambda) = - \frac{P_{\xi_1}d\xi_1}{P_\lambda} - \frac{P_{\xi_2}d\xi_2}{P_\lambda} = P_{\xi_1}d\xi_1 + P_{\xi_2}d\xi_2 	
$$
If we wedge this form with 
$$
\Omega_\lambda = \frac{1}{ (2\pi i)^2 }\frac{d \xi_1 }{ \xi_1 \xi_2 P_{\xi_2}}    
$$
then
\begin{equation}
\Omega_\lambda \wedge \pi^*(d\lambda) = \frac{1}{ (2\pi i)^2 } \frac{d \xi_1  d \xi_2  }{  \xi_1 \xi_2 }                             
\end{equation}

This observation was already made in $ \S 11 $ of \cite{GKT}. Note that $\Omega_\lambda$ is slightly different from $\omega_\lambda$. As defined before, $\omega_\lambda$ is a volume form so it should be positive and we finally got a positive period. By contrast, as for  $\Omega_\lambda$, we get rid of the absolute value and take the orientation into account. 

Define $\tilde{\Omega} : =  \Omega_\lambda \wedge \pi^*d\lambda$ over $B_{m,n}^{(k,\ell)}$ such that $\tilde{\Omega}|_{F^{k,\ell,m,n}_\lambda} = \Omega_\lambda$, where $F^{k,\ell,m,n}_\lambda$ is the component of the Fermi curve of $B_{m,n}^{(k,\ell)}$, $\tilde{\Omega}$ is called the density of states form and then $\Omega_\lambda$  is the relative differential form with respect to $\pi^*d\lambda$. 
$$ 
\begin{array}{l @{\quad} l} 
\int_{B_{m,n}^{(k,\ell)}}  |\tilde{\Omega}| &  = \int_{B_{m,n}^{(k,\ell)}} |\Omega_\lambda \wedge \pi^*d\lambda|  \\
     &    = \int_{0 < |\lambda| < 4} d\lambda \int_{F^{k,\ell,m,n}_\lambda} \omega_\lambda \\
     &    = \frac{1}{4\pi^2 ab} \int_{0 < |\lambda| < 4} d\lambda (1+k)K(k)  \\
     &    = \frac{1}{2\pi^2 ab}  \int_0^1 (1+k)K(k) \frac{8}{(1+k)^2} dk \\
     &    =  \frac{4}{\pi^2 ab} \int_0^1 \frac{K(k)}{1+k } dk  \\
     &    = \frac{1}{2ab}                        
\end{array}
$$

Then it is possible to construct some interesting spectral functions based on $d\rho$ for each component. 
Let us look at the zeta function of the Harper operator on $B_{m,n}^{(k,\ell)}$, 
\begin{equation}
\zeta_H^{k,\ell,m,n}(s):=\int_{0 < |\lambda| < 4} \lambda^s d\rho =  \int_{B_{m,n}^{(k,\ell)}}  \lambda^s |\tilde{\Omega}|            
\end{equation}

\begin{equation} 
\begin{array}{l @{\quad} l} 
\int_{0 < |\lambda| < 4} \lambda^s d\rho & = \int_0^4   \lambda^s d\lambda \int_{F^{k,\ell,m,n}_\lambda} \omega_\lambda +\int_{-4}^0  \lambda^s d\lambda \int_{F^{k,\ell,m,n}_\lambda} \omega_\lambda\\
     &    = \frac{2^{2s+1}}{\pi^2 ab} \int_0^1 (\frac{1-k}{1+k})^s \frac{K(k)}{1+k} dk + \frac{2^{2s+1}}{\pi^2 ab} \int_0^1 (\frac{k-1}{1+k})^s \frac{K(k)}{1+k} dk             
\end{array}
\end{equation}

In particular, when $s = 2k+1$ is an odd integer, $\zeta_H^{k,\ell,m,n}(2k+1) = 0$ and when $s = 2k$ is even, 
$\zeta_H^{k,\ell,m,n}(2k) = \frac{4^{2k+1}}{\pi^2 ab} \int_0^1 (\frac{1-k}{1+k})^{2k} \frac{K(k)}{1+k} dk $.

There is another way to compute the zeta function using two dimensional residue theorem. 
Denote 
$$
W(\xi_1, \xi_2):= e^{2\pi i \alpha(n + \ell b)}\xi_1 + e^{-2\pi i \alpha(n +\ell b)  } \xi_1^{-1} + e^{ 2\pi i \beta( m +ka)  }\xi_2 + e^{-2\pi i \beta( m +ka)  } \xi_2^{-1}
$$ 
namely, $P(\xi_1, \xi_2, \lambda) =   W(\xi_1, \xi_2) - \lambda$. Since for any function $f(\xi_1, \xi_2, \lambda)$,
\begin{tabbing}
$\qquad \qquad \qquad \qquad \int_{\mathbb{T}^2, 0< |\lambda| < 4} f(\xi_1, \xi_2, \lambda) \delta(\lambda - W(\xi_1, \xi_2)) d\xi_1 d\xi_2 d\lambda$ \\
$\qquad \qquad \qquad \qquad = \int_{\mathbb{T}^2, 0< |\lambda| < 4, \lambda = W} f(\xi_1, \xi_2, \lambda) \frac{d\sigma}{|\nabla P|}$ \\
$ \qquad \qquad \qquad \qquad = \int_{\mathbb{T}^2, 0< |\lambda| < 4, \lambda = W} f(\xi_1, \xi_2, \lambda) \frac{d\xi_1d\xi_2 + d\xi_1d\lambda + d\xi_2d\lambda}{\sqrt{1+P^2_{\xi_1}+ P^2_{\xi_2}}}$ \\
$\qquad \qquad \qquad \qquad = \int_{B_{m,n}^{(k,\ell)}} f(\xi_1, \xi_2, \lambda) \frac{1+P_{\xi_1}+ P_{\xi_2}}{\sqrt{1+P^2_{\xi_1}+ P^2_{\xi_2}}} d\xi_1d\xi_2$ 
\end{tabbing}
Therefore 
\begin{tabbing}
$ \qquad \qquad \int_{B_{m,n}^{(k,\ell)}} \lambda^s  \frac{d\xi_1d\xi_2}{4\pi^2 \xi_1 \xi_2}$ 
$= \int_{\mathbb{T}^2, 0< |\lambda| < 4} \lambda^s  \delta(\lambda - W)  \frac{\sqrt{1+P^2_{\xi_1}+ P^2_{\xi_2}}}{1+P_{\xi_1}+ P_{\xi_2}} \frac{d\xi_1d\xi_2 d\lambda}{4\pi^2 \xi_1 \xi_2}$  \\
$ \qquad \qquad \qquad \qquad ~~~~~~= \int_{\mathbb{T}^2} W^s  \frac{\sqrt{1+P^2_{\xi_1}+ P^2_{\xi_2}}}{1+P_{\xi_1}+ P_{\xi_2}} \frac{d\xi_1d\xi_2 }{4\pi^2 \xi_1 \xi_2}$ 
\end{tabbing}
By the residue theorem, only the positive integer powers survive and the contour integral only depends on terms with $1, \xi_1, \xi_2, \mbox{and} ~\xi_1\xi_2$ in the integrand.

We can also have the partition function of the Harper operator on $B_{m,n}^{(k,\ell)}$, 
$$
Z_H^{k,\ell,m,n}(t):=\int_{0 < |\lambda| < 4} e^{-t \lambda} d\rho = \frac{1}{2ab} + \sum_{k=1}^{\infty}\frac{\zeta_H^{k,\ell,m,n}(2k)}{(2k)!}t^{2k}
$$
In order to get the spectral functions on the whole variety $B$, we should collect all the contributions from countable components.

\section{Almost Mathieu Operator}\label{DenSec2}

We compare in this section the case of the density of states of the two-dimensional
Harper operator analyzed above with the analogous problem for the one-dimensional
almost Mathieu operator. In particular, we show how to recover in the one-dimensional
case the familiar picture of the Hofstadter butterfly and the corresponding 
density of states with its explicit dependence on the parameter. 

\subsection{Algebro-geometric model}

We apply the same process to the almost Mathieu operator, the big difference is that now its density of states depends on the parameter $\alpha$. Therefore the derived spectral functions will have totally different properties compared with those of the Harper operator. Here again we have the analog of $\S 3$ of \cite{GKT}.

First the Bloch variety of the almost Mathieu operator is now given by
\begin{equation}
B' := \{ (\xi, \lambda) \in \mathbb{C}^* \times \mathbb{C}
 ~|~ H' \varphi(n) =\lambda \varphi(n), \varphi(n+a) = \xi \varphi(n)\}
\end{equation}
It is easy to see that $(\xi, \lambda)$ belongs to $B'$ if and only if $\lambda$ is an eigenvalue of the matrix $M'$ with countable components $M'^{\ell}$, $\ell \in \mathbb{Z}$
\begin{equation}
M'^{\ell} = \left(
\begin{array}{ccccc}
2cos(2\pi \alpha(1+\ell a)) & 1                   & 0       & \ldots & -\xi \\
1                 & 2cos(2\pi \alpha(2+\ell a) ) & 1       & \ldots & 0\\
\vdots            & \vdots              & \ddots  & \vdots & \vdots \\
-\xi^{-1}         & 0                   &\ldots   & 1        & 2cos(2\pi \alpha (\ell +1)a)
\end{array} \right)
\end{equation}
In other words, 
\begin{equation}
B'=\{(\xi, \lambda)  \in \mathbb{C}^* \times \mathbb{C}  ~|~ det(M' - 
\lambda I) = \prod_{\ell \in \mathbb{Z}}det(M'^{\ell} - \lambda I) = 0\}
\end{equation}
In fact we can expand the determinant as $det(M'^{\ell} - \lambda I) = p_{\ell}(\lambda, \alpha) - \xi - \xi^{-1}$,  
\begin{equation}
p_{\ell}(\lambda, \alpha) := (-\lambda)^a + (2\sum_{j=1}^a cos2\pi  \alpha (j+\ell a) ) (-\lambda)^{a-1}+ O(\lambda^{a-2})
\end{equation}

Consider a continued fraction expansion $\{ c_n = \frac{p_n}{q_n} \}$ to approximate the irrational $\alpha$. Then we let the prime $a$ vary according to $\{ q_n \} $ so that the spectrum of the Bloch variety can be approximated by numerical computation of finite band structures. Indeed, based on semi-classical analysis and renormalization method, the Hofstadter butterfly is obtained by numerical approximations \cite{B}. 

Define the unramified covering and  $\tilde{B'} : = c'^{-1}(B')$ as before
$$
\begin{array}{l}
   c':\mathbb{C}^* \times  \mathbb{C} \rightarrow \mathbb{C}^* \times  \mathbb{C} \\
     ~~~~~ (z, \lambda) ~\mapsto ~(z^a, \lambda)
\end{array}
$$
The structure group $\mu_a$ of the covering $c': \tilde{B'} \rightarrow B'$ acts on the fibers as 
$$
\rho \cdot (z, \lambda) = (\rho z, \lambda)
$$
Similarly, we use Fourier transform to convert the spectrum problem into the study of the Fourier modes, then the almost Mathieu operator
is represented by the infinite matrix
\begin{equation}
\hat{M'} = Diag( \rho z +  \rho^{-1} z^{-1} +2cos2\pi \alpha (j + \ell a) ), \quad 1 \leq j \leq a, ~\ell \in \mathbb{Z} 
\end{equation}
So under this representation, the lifted Bloch variety 
\begin{equation}
\tilde{B'} = \{ (z, \lambda)
 ~| ~   \prod_{\ell \in \mathbb{Z}}\prod_{\rho \in \mu_a}\prod_{j=1}^a( \rho z +  \rho^{-1} z^{-1} +2cos(2\pi \alpha (j + \ell a) ) -\lambda) = 0 \}
\end{equation}
and the Bloch variety 
\begin{equation}
B' = \{ (\xi, \lambda)
 ~| ~ \prod_{\ell \in \mathbb{Z}} \prod_{j=1}^a ( \xi  +  \xi^{-1}  +2cos2\pi \alpha (j + \ell a)  -\lambda) = 0 \}
\end{equation}
 In this case, the Fermi curves are degenerate points
\begin{equation}
F'_\lambda = \bigcup_{\ell \in \mathbb{Z}} \bigcup_{j=1}^a\{ \xi  ~| ~   \xi  +  \xi^{-1}  +2cos2\pi \alpha (j + \ell a ) = \lambda  \}
\end{equation}

\subsection{Density of States}
Let us look at the self-adjoint boundary value problem in dimension one.  For $k \in I$, we now consider 
\begin{equation}
B' := \{ (e^{2\pi i k}, \lambda) 
 ~|~ H' \varphi(n) =\lambda \varphi(n), \varphi(n+a) = e^{2\pi i k} \varphi(n)\}
\end{equation}
Assume its band functions consist of $\{ E_i(k)\}$. 

Let $H'_n(n \geq 1)$ denote the self-adjoint operator $H'$ acting on $l^2(\mathbb{Z}/an\mathbb{Z})$.
For $k \in \frac{1}{n}\mathbb{Z}$, $E_i(k)$ is an eigenvalue of $H'_n$ and the spectrum of $H'_n$ is
$$
\{ E_i(\frac{m}{n} ) | i \geq 1, ~ 1 \leq m \leq n \}.
$$ 
Similarly we have 
$$
\nu_n(\lambda) = \sum_{i=1}^{\infty} \sum_{m= 1}^n \Theta(\lambda - E_i(\frac{m}{n}))
$$
Then the integrated density of states of the almost Mathieu operator is 
\begin{equation}
\begin{array}{l @{\quad} l}
\rho'(\lambda) &  =  \lim_{n \to \infty}\frac{1}{a} \sum_{i=1}^{\infty} \frac{1}{n}\sum_{m = 1}^n \Theta(\lambda - E_i(\frac{m}{n} ))\\
              &  =  \frac{1}{ a} \sum_{i=1}^{\infty} \int_{I} \Theta(\lambda - E_i(k))dk
\end{array}
\end{equation}
And the density of states of the almost Mathieu operator is 
\begin{equation}
\frac{ d \rho'  }{ d \lambda }=  \frac{1}{ a} \sum_{i=1}^{\infty} \int_{ I} \delta(\lambda - E_i(k) ) dk
                              =  \frac{1}{ a} \sum_{i=1}^{\infty}  \frac{ 1}{ | g_i'(k_{i}) | }
\end{equation}
where $g_i(k) := \lambda - E_i(k)$ and $k_{i}$ are the real zeros of $g_i(k)$. From the Fermi curve of the almost Mathieu operator, we have $ E_i(k) =  E_j^\ell(k) $ for some $ j, \ell$
\begin{equation}
E_j^l(k)  = 2cos 2 \pi k + 2 cos 2 \pi \alpha  (j + \ell a) = \lambda
\end{equation}
i.e. $k_{i} \in I$  satisfies the equation $2cos 2 \pi k_{i} + 2 cos 2 \pi \alpha  (j + \ell a) =\lambda $. It is easy to see that 
$$
g_i'(k) = -\frac{d}{dk}E_j^\ell(k) =  4\pi sin 2\pi k
$$
Hence
$$
\frac{d\rho'}{d\lambda} =  \frac{1}{ a} \sum_{i=1}^{\infty} \frac{ 1}{ 4\pi | sin 2\pi k_{i} | } 
                        =  \frac{1}{ 4\pi a} \sum_{\ell \in \mathbb{Z}}\sum_{j=1}^{a}  \frac{ 1}{\sqrt{1 - (\lambda/2 - cos 2 \pi \alpha (j+\ell a))^2} }
$$

Therefore the density of states of the almost Mathieu operator is a function of the parameter $\alpha$. However, as we have seen that the density of states of the Harper operator is independent of the parameters $\alpha$ and $\beta$.

In dimension two, the integral variable was absorbed by integrating over one-dimensional Fermi curves and the dependence of the parameters was resolved by the symmetric form of the magnetic translations. But for the almost Mathieu operator, the Fermi curves are degenerate points and the associated measure is just the counting measure, so the dependence of the parameter still remains, 
as expected, in agreement with the form of the density described for instance in  \cite{HKS}.

\subsection{Spectral functions}
Now we can similarly define the zeta function of the almost Mathieu operator on each component $B'_{\ell ,j}$
\begin{equation} 
\zeta_{H'}^{\ell ,j}(s) : =\int_{B'_{\ell ,j}} \lambda^s \frac{d\xi}{2\pi i \xi}   = \int_{\mathbb{T}} W_{\ell,j}'^s \frac{d\xi }{2\pi i \xi}
                      = \int_{\mathbb{T}} (\xi + \xi^{-1} + 2 cos 2 \pi \alpha (j+\ell a))^s \frac{d\xi }{2\pi i \xi}
\end{equation}

By the residue theorem, only the positive integer powers survive and the contour integral only depends on terms with $1~ \mbox{and} ~\xi$. So consider powers $ n \in \mathbb{N}$ and expand the polynomial as
$$
(\xi + \xi^{-1} + 2 cos 2 \pi \alpha (j+\ell a))^n = \sum_{k_1, k_2, k_3}C_n^{k_1, k_2, k_3}\xi^{k_1-k_2} (2 cos 2 \pi \alpha (j+\ell a))^{k_3}
$$
There are only two cases 
$$
(1)~k_1 -k_2 = 0 ~~\mbox{and} ~~k_1 + k_2 +k_3 = n~~~~
(2)~k_1 -k_2 = 1 ~~\mbox{and} ~~k_1 + k_2 +k_3 = n~~~~
$$
thus $\zeta_{H'}^{\ell,j}(n)$ can be computed as `winding numbers'
$$
\zeta_{H'}^{\ell,j}(n)  = \sum_{n-2k \geq 0} \frac{n!(2 cos 2 \pi \alpha(j+\ell a))^{n-2k}}{(k!)^2(n-2k)!} + \frac{1}{2\pi i}\sum_{n-2k-1 \geq 0} \frac{n!(2 cos 2 \pi \alpha (j+\ell a))^{n-2k-1}}{k!(k+1)!(n-2k-1)!}
$$

To get the zeta function of the whole variety $B'$, we have $\zeta_{H'}(n) =  \sum_{\ell \in \mathbb{Z}} \sum_{j=1}^a\zeta_{H'}^{\ell,j}(n)$. Finally, we also have the formal partition function of the almost Mathieu operator on $B'$
$$
Z_{H'}(t):=\int e^{-t \lambda} d\rho' = \sum_{k=0}^{\infty}\frac{\zeta_{H'}(k)}{k!}(-t)^{k}
$$

\section{Conclusions}

In the case of a discretized periodic Schr\"odinger operator describing electron
propagation in solids, the complex energy-crystal momentum dispersion relation
is described geometrically by an algebraic variety, the Bloch variety, which 
consists of the set of complex points that can be reached by analytic continuation
of the band functions. The density of states can then be computed \cite{GKT} as
a period on a curve (the Fermi curve) in the Bloch variety.

\smallskip

In this paper we consider the case with magnetic field, where the periodic Schr\"odinger
operator is replaced by the two-dimensional Harper operator, or its degeneration, the
one-dimensional almost Mathieu operator. One knows from the spectral theory of these
operators (\cite{HKS}, \cite{L1}, \cite{Sh}) that, in the case of irrational parameters,
the band structure of the spectrum is replaced by a Cantor set, giving rise to the
well known Hofstadter butterfly picture. Thus, one can see that, correspondingly,
the geometric locus describing the complex energy-crystal momentum dispersion relation
and replacing the Bloch variety will no longer be directly described by
classical algebraic geometry. However, as one approximates the Cantor-like spectrum
by a family of intervals (for example by approximating the irrational parameter by rationals
via the continued fraction algorithm), it should be possible to correspondingly ``approximate"
this geometric space by ordinary algebraic varieties. The first part of this paper consists
of a geometric result, which shows exactly what this approximation and limit procedure
consists of. In particular, we show that one can obtain the space describing the
complex energy-crystal momentum dispersion relation for the Harper operator with
irrational parameters as a ``double limit" of a family of algebraic varieties, or more
precisely as an ind-pro-variety, where one limit takes care of the presence of infinitely
many components and the other limit of the blowups that are needed to deal with the
singularities. The resulting space has indeed a Cantor-like geometry, but one that admits
a good approximation by algebraic varieties, in the sense of this double limit procedure,
so that one can still use methods from classical algebraic geometry, applied 
compatibly to the varieties in the approximating family.

\smallskip

We then use this geometric result to show that we can still obtain an explicit
calculation of the density of states for the two-dimensional Harper operator 
as a period on the Fermi curve, where the 
period integral now consists of a sequence of compatible contributions from the
components of the approximating system of varieties. 
This integral is then explicitly computed in terms of elliptic
integrals and periods of elliptic curves. Similarly, we obtain explicit formulae
for the spectral functions, again in terms of compatible contributions from
the approximating family of algebraic varieties. We apply the same technique
to the density of states and spectral functions in the case of the one-dimensional 
almost Mathieu operator. The main difference
between the two cases is that, in the two-dimensional case the dependence on
the parameter disappears in the density of states, because it is absorbed in the integration 
over the Fermi curve, while in the one-dimensional case the Fermi curves are
points with the counting measure and one recovers the density of states
obtained, by different methods, in \cite{HKS}, with its explicit dependence 
on the parameter. By a residue calculation, we also obtain the zeta function 
of the almost Mathieu operator as a sum over ``winding numbers" associated 
to the components in the approximating family of algebraic Bloch varieties.

\end{document}